\documentclass[sigconf,natbib=false]{acmart}
\AtBeginDocument{%
  }

\setcopyright{acmlicensed}
\copyrightyear{2026}
\acmYear{2026}
\acmDOI{XXXXXXX.XXXXXXX}
\acmConference[ACM MobiSys]{The 24th Annual International Conference on Mobile Systems, Applications, and Services}{June 21--25, 2026}{Cambridge, UK}
\acmISBN{978-1-4503-XXXX-X/2018/06}

\RequirePackage[
  datamodel=acmdatamodel,
  style=acmnumeric,
  sorting=none,
  ]{biblatex}

\usepackage{tcolorbox}
\newtcolorbox{summarybox}{
    colback=gray!10, colframe=black, boxrule=0.5pt, arc=2pt,
    left=4pt, right=4pt, top=4pt, bottom=4pt, boxsep=0pt
}
\usepackage{graphicx}
\usepackage{amsthm}
\usepackage{amsmath}

\usepackage{amssymb}
\usepackage{multirow}
\usepackage{pgfplots}
\pgfplotsset{compat=1.18}
\usepackage{booktabs}
\usepackage{graphicx} 
\usepackage{tikz}
\usepackage{fontawesome5} 
\usetikzlibrary{shadows, positioning, arrows.meta}
\usetikzlibrary{shadows, positioning, arrows.meta, shapes.symbols, fit, backgrounds, calc, shapes.geometric, decorations.pathmorphing}
\usepackage[table]{xcolor} 
\usepackage{enumitem}

\usepackage{tcolorbox}
\tcbuselibrary{skins, breakable}

\newtcolorbox{rqbox}{
  colback=gray!8,
  colframe=black!60,
  fonttitle=\bfseries,
  arc=1mm,                    
  boxrule=1.5pt,
  left=5pt,
  right=5pt,
  top=5pt,
  bottom=5pt
}

\newtcolorbox{answerbox}{
  colback=gray!12,
  colframe=black!70,
  fonttitle=\bfseries,
  arc=1mm,                    
  boxrule=1.5pt,
  left=5pt,
  right=5pt,
  top=5pt,
  bottom=5pt
}

\definecolor{ssgreen}{RGB}{27, 158, 119}  
\definecolor{ssred}{RGB}{217, 95, 2}      
\definecolor{ssorange}{RGB}{230, 171, 2}  

\definecolor{ssgreen}{RGB}{27, 158, 119} 
\definecolor{ssred}{RGB}{217, 95, 2}     
\definecolor{ssblue}{RGB}{117, 112, 179} 
\begin{document}

\acmYear{2026}\copyrightyear{2026}
\setcopyright{cc}
\setcctype[4.0]{by}
\acmConference[MobiSys '26]{The 24th Annual International Conference on Mobile Systems, Applications and Services}{June 21--25, 2026}{Cambridge, United Kingdom}
\acmBooktitle{The 24th Annual International Conference on Mobile Systems, Applications and Services (MobiSys '26), June 21--25, 2026, Cambridge, United Kingdom}
\acmDOI{10.1145/3745756.3809189}
\acmISBN{979-8-4007-2027-7/26/06}

\title{StreamSplit: Continuous Audio Representation Learning via Uncertainty-Guided Adaptive Splitting}


\author{Minh K. Quan}
\affiliation{%
  \institution{School of Engineering, Deakin University}
  \city{Waurn Ponds}
  \country{Australia}}
\email{m.quan@deakin.edu.au}

\author{Pubudu N. Pathirana}
\affiliation{%
  \institution{School of Engineering, Deakin University}
  \city{Waurn Ponds}
  \country{Australia}}
\email{pubudu.pathirana@deakin.edu.au}

\renewcommand{\shortauthors}{Quan and Pathirana}

\begin{abstract}
  Large-batch Contrastive Learning (CL), the foundation of modern representation learning, is fundamentally incompatible with the volatile resource constraints of edge devices. This conflict creates a dilemma: small on-device batches degrade model fidelity, while offloading to the cloud incurs unacceptable latency and bandwidth costs. Existing solutions often resort to static model compression, which fails to adapt to the runtime volatility of edge environments. To bridge this gap, we present StreamSplit, a novel framework that makes streaming CL practical across heterogeneous ARM client platforms. StreamSplit resolves the conflict between the \textbf{continuous nature of ambient audio} and the discrete batch requirements of models like CLAP and COLA. We introduce: (1) A distribution-based streaming framework that decouples representation quality from local batch size, using a tractable Hybrid Loss to maintain fidelity despite sparse updates; and (2) An Uncertainty-Guided Adaptive Splitter that uses a lightweight Reinforcement Learning (RL) policy to dynamically partition computation. Uniquely, this policy integrates real-time resource monitoring with embedding ambiguity to optimize the accuracy-latency trade-off on the fly. We evaluate StreamSplit on diverse hardware, from the resource-constrained Raspberry Pi 4 to the high-performance Apple M2. Results demonstrate that StreamSplit reduces per-sample latency by up to $4.7\times$ and cuts bandwidth by 77.1\% and energy by 52.3\% compared to server-centric baselines. Crucially, it maintains accuracy within 2.2\% of \textbf{server-centric} models, proving that adaptive, distributed learning is a viable path for the modern edge ecosystem.
\end{abstract}




\maketitle

\section{Introduction}
\label{sec:intro}


Audio-capable edge devices, from smart speakers and wearables to mobile robots, offer immense potential for ambient intelligence \cite{Angel2022, Li2024ambient}. To realize this vision, these devices must learn from streaming data and adapt to complex acoustic environments, such as performing robust scene classification or source separation in smart cities \cite{quan2025quantum, quan2025quantum_source}. However, the state-of-the-art models for audio representation learning (e.g., CLAP \cite{Wu2023}, COLA \cite{Saeed2021}, CLAR \cite{al2021clar}) are largely confined to data centers. This centralization creates a critical bottleneck: transmitting continuous raw audio to the cloud is unsustainable regarding bandwidth \cite{Alharbi2023, Zhang2025edge} and energy \cite{Sakip2023}. 

Furthermore, transmitting raw user audio raises \textbf{data minimization concerns}; processing data locally or transmitting only intermediate representations aligns better with data minimization principles, though it does not replace cryptographic guarantees.


\begin{figure}[t]
\centering
\includegraphics[width=0.99\columnwidth]{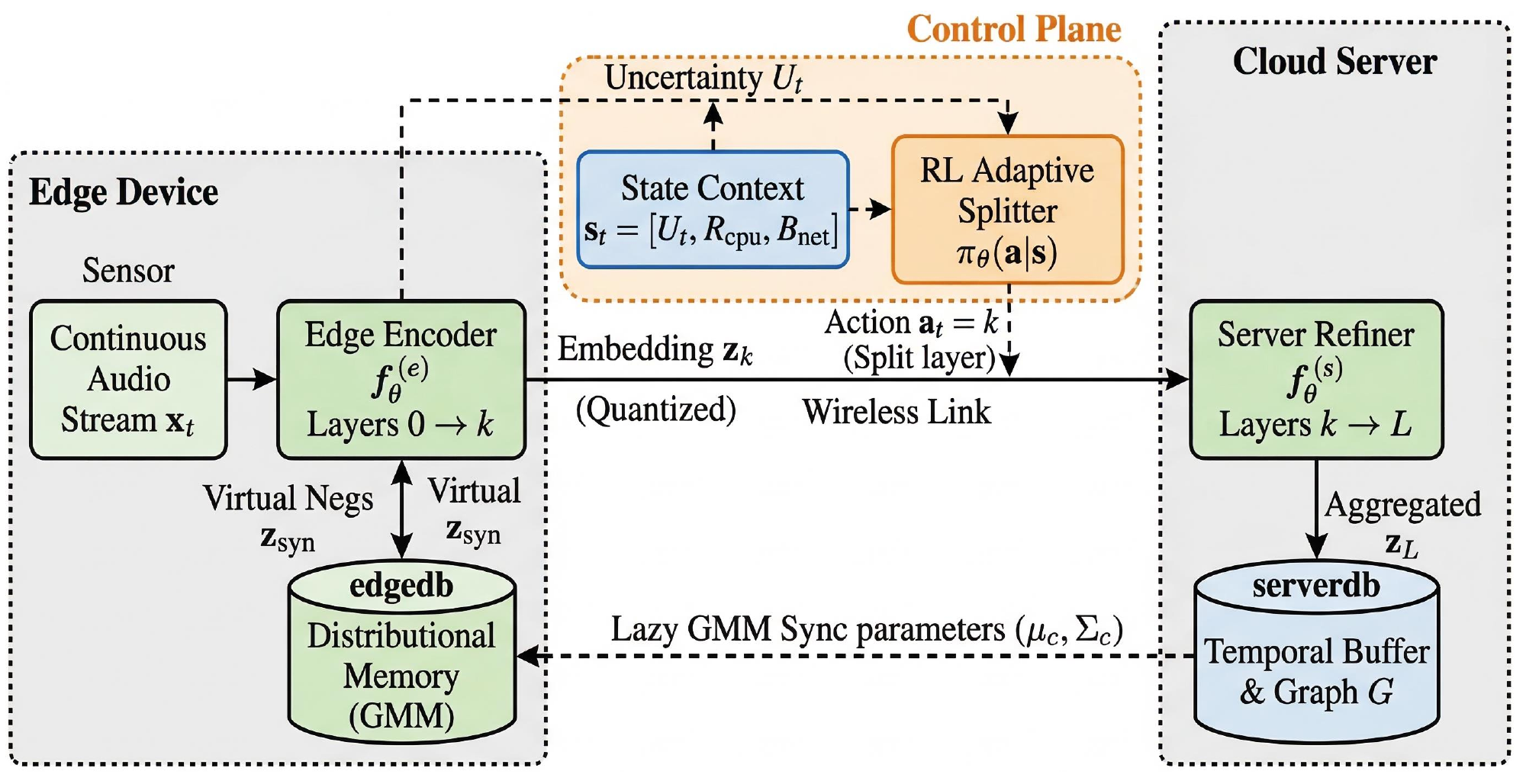}
\caption{\textbf{StreamSplit System Architecture.} The framework comprises three coupled subsystems: an Edge Learner optimizing local throughput via virtual negatives, a Control Plane dynamically partitioning the network based on state constraints, and a Server Refiner maintaining manifold continuity via Laplacian regularization on graph $\mathcal{G}$.}
\label{fig:architecture}
\end{figure}

Edge computing offers a solution by moving computation closer to the source \cite{Lin2019, Angel2022}. Yet, simply porting server-grade training algorithms to the edge introduces fundamental conflicts. Contrastive Learning (CL)—the dominant paradigm for self-supervised audio learning—relies on \textbf{large batches} and \textbf{diverse negative samples} to learn high-quality features \cite{Chen2020, He2020, Grill2020}. This requirement is inherently at odds with the limited memory and compute of client devices \cite{Akherfi2018, Akhlaqi2023}. Furthermore, the edge ecosystem is \textbf{heterogeneous} and \textbf{volatile}; a static deployment that works on a high-performance Apple M2 chip may drain the battery of a Raspberry Pi, and a partition optimized for strong WiFi may paralyze the application when the network degrades \cite{Arcas2024, Zhou2023computing}.

To enable practical, streaming learning at the edge, we must address two distinct, coupled challenges:

\textbf{Challenge 1: The Stream-Clip Mismatch (Algorithmic).} 
State-of-the-art audio models (e.g., CLAP) differ fundamentally from edge sensors: they are trained on shuffled, discrete \textit{clips} (files), whereas edge devices perceive continuous \textit{streams}. Naively chopping streams into small batches destroys the \textbf{temporal coherence} required to learn robust features. We need a formulation that exploits acoustic continuity rather than fighting it.

\textbf{Challenge 2: The Volatility Conflict (Systems).} 
The optimal point to split a neural network between a device and a server is not static. It shifts constantly based on the device's instantaneous CPU load (e.g., background tasks) and external network conditions. Static partitioning schemes or "train-then-compress" approaches fail to adapt \cite{Thapa2022, He2023}. In our preliminary analysis, we observed that a static split (Fixed Split Learning) suffered a \textbf{15.7\% accuracy drop} when subjected to realistic CPU load volatility, due to latency timeouts and dropped frames.

Existing paradigms isolate representation learning from system scheduling. Standard split computing relies blindly on system telemetry \cite{Kang2017}, while small-batch continual learning demands large memory queues that cause edge cache thrashing \cite{Fang2022}. StreamSplit avoids this via Distributional Memory, synthesizing virtual negatives from a compact Gaussian Mixture Model (GMM) to prevent dimensional collapse without large physical batches. Furthermore, naively pairing heuristic schedulers with small-batch methods fails for continuous audio: dropping a semantically "hard" frame shatters the embedding manifold. StreamSplit overcomes this via \textit{algorithm-system co-design} \cite{Huang2023roof}. By leveraging the GMM's entropy as a lightweight, zero-cost uncertainty signal, our RL routing policy directly couples algorithmic difficulty to system execution, intelligently spending bandwidth only where it maximizes server utility.

In this paper, we introduce \textbf{StreamSplit} (Figure \ref{fig:architecture}), a cohesive framework to resolve conflicts. Unlike static compression techniques \cite{Huang2023roof}, StreamSplit co-designs the learning algorithm with the system architecture for robust, adaptive execution.

To address \textbf{C1}, we propose a \textbf{distribution-based streaming framework}. Instead of relying on large discrete batches, StreamSplit aligns the \textit{distribution} of edge embeddings with a global prior using a hybrid Sliced-Wasserstein and Laplacian loss \cite{Kolouri2019, Belkin2006}. This allows the edge device to contribute high-quality updates using small, local batches, effectively decoupling representation quality from on-device memory constraints.

To address \textbf{C2}, we introduce an \textbf{Uncertainty-Guided Adaptive Splitter}. We formulate the splitting decision as a lightweight RL problem. Our agent monitors system metrics (CPU, RAM, Network) and—crucially—the \textit{uncertainty} of the current audio embedding. This allows StreamSplit to offload "hard" samples to the server while processing "easy" samples locally, or to aggressively compress data when the network is congested.

We implement StreamSplit on heterogeneous ARM platforms from the resource-constrained \textbf{Raspberry Pi 4B} to the high-performance \textbf{Apple MacBook M2}. Our contributions include: (1) the first framework integrating distribution-based contrastive learning with RL-based system control for both convergence and stability; (2) platform-agnostic deployment extending battery life by 50\% on Pi 4 while leveraging NPU acceleration on M2; (3) 77.1\% bandwidth and 52.3\% energy reduction while maintaining accuracy within 2.2\% of server models.


\section{Background and Motivation}
\label{sec:background}

To understand the necessity of StreamSplit, we must examine the fundamental conflicts that render current state-of-the-art approaches insufficient for the edge: the algorithmic dependency of contrastive learning on massive data batches, and the systemic intolerance of split computing to resource volatility. Table \ref{tab:gap_analysis} provides a systematic critique of these paradigms.

\subsection{Contrastive Learning: The Batch Size Barrier}
Contrastive Learning (CL) has established itself as the dominant paradigm for self-supervised audio representation, driving recent breakthroughs in ambient intelligence \cite{Angel2022, Li2024ambient}. State-of-the-art models like COLA \cite{Saeed2021} and CLAP \cite{Wu2023} rely on maximizing agreement between augmented views of data to learn robust features without human labels.

\textbf{Dependency on Negative Diversity.} The efficacy of CL is mathematically rooted in the quantity and diversity of negative samples available during a gradient update. The standard InfoNCE loss function \cite{Chen2020} relies on the denominator to approximate the true data distribution. To maintain this approximation, leading frameworks like SimCLR \cite{Chen2020} and MoCo \cite{He2020} require massive batch sizes (often $N > 2048$) or large memory queues. Without sufficient diversity, the gradients become biased, and the learning objective degenerates \cite{Wang2020}.

\textbf{The Edge Incompatibility (C1).} This requirement creates fundamental friction with edge hardware. Embedded devices, constrained by strict memory budgets \cite{Akherfi2018, Akhlaqi2023}, often restrict batch sizes to single digits ($N=8$ or $16$). Existing adaptations fall short: gradient accumulation introduces prohibitive latency for streaming \cite{Kong2023}, while small-batch training causes dimensional collapse where embeddings map to narrow cones \cite{Grill2020}. StreamSplit addresses this by shifting from sample-based to \textit{distribution-based} alignment (Section \ref{sec:framework}), decoupling representation quality from local batch size.

\subsection{Split Computing: The Volatility Gap}
To circumvent on-device resource limitations, Split Computing partitions the neural network between the edge and cloud \cite{Lin2019}. While theoretically sound, existing implementations fail to address the dynamic nature of mobile environments.

\begin{figure}[t]
\centering
\includegraphics[width=0.99\columnwidth]{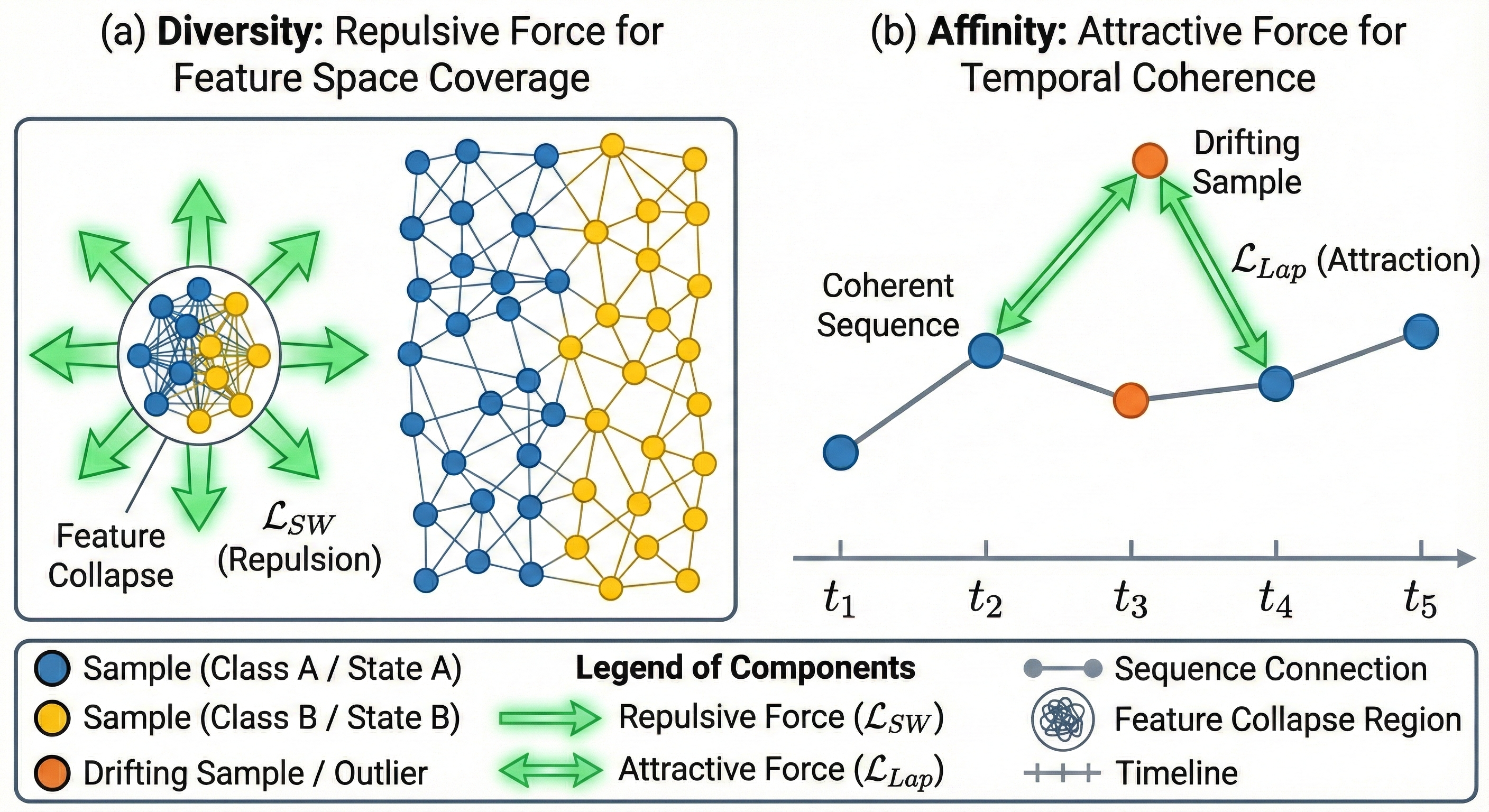}
\caption{\textbf{Embedding Quality Metrics.} \textbf{(a) Diversity ($\mathcal{L}_{SW}$):} Repulsive force preventing dimensional collapse by dispersing embeddings. \textbf{(b) Affinity ($\mathcal{L}_{Lap}$):} Attractive force ensuring temporal coherence by pulling drifting frames (red) back to the manifold.}
\label{fig:metrics_concept}
\end{figure}

\begin{table*}[t]
  \centering
  \footnotesize
  \caption{\textbf{Comprehensive Gap Analysis.} Current paradigms fail to simultaneously satisfy the conflicting requirements of the edge. While heuristic methods attempt adaptation, they lack data awareness (Uncertainty), leading to inefficient offloading decisions.}
  \label{tab:gap_analysis}
  \centering
  \resizebox{\textwidth}{!}{%
  \begin{tabular}{lcccccc}
    \toprule
    \multirow{2}{*}{\textbf{Paradigm}} & \multirow{2}{*}{\textbf{Representative Works}} & \multicolumn{2}{c}{\textbf{Methodology}} & \multirow{2}{*}{\textbf{Critical Failure Mode}} & \multicolumn{2}{c}{\textbf{Key Properties}} \\
    \cmidrule(lr){3-4} \cmidrule(lr){6-7}
    & & \textbf{Algorithm} & \textbf{Execution} & & \textbf{Quality} & \textbf{Adaptive} \\
    \midrule
    \textbf{Cloud-Centric} & CLAP \cite{Wu2023}, COLA \cite{Saeed2021} & Large-Batch & Full Offload & Privacy \& Bandwidth Costs & \textcolor{ssgreen}{\faCheck} & -- \\
    \textbf{Edge-Only} & EdgeNCE \cite{Wang2022edgence}, MCUNet & Small-Batch & On-Device & Dimensional Collapse & \textcolor{ssred}{\faTimes} & \textcolor{ssred}{\faTimes} \\
    \textbf{Static Split} & SplitFed \cite{Thapa2022}, FedSL \cite{He2023} & Standard & Fixed Layer & Latency Timeouts (-15.7\%) & \textcolor{ssgreen}{\faCheck} & \textcolor{ssred}{\faTimes} \\
    \textbf{Rule-Based Split} & RoofSplit \cite{Huang2023roof}, Neurosurgeon & Standard & Resource-Aware & \textbf{Optimization Blindness} & \textcolor{ssgreen}{\faCheck} & \textcolor{ssorange}{\faExclamationTriangle} \\
    \midrule
    \textbf{StreamSplit (Ours)} & -- & \textbf{Distributional} & \textbf{Uncertainty RL} & \textbf{Robust (Co-Designed)} & \textcolor{ssgreen}{\textbf{\faCheck}} & \textcolor{ssgreen}{\textbf{\faCheck}} \\
    \bottomrule
    \multicolumn{7}{l}{\footnotesize \textcolor{ssgreen}{\faCheck} = Fully Supported; \textcolor{ssred}{\faTimes} = Fails/Not Supported; \textcolor{ssorange}{\faExclamationTriangle} = Unreliable or Data-Agnostic (Partial Support); -- = Not Applicable.}
  \end{tabular}
  }
\end{table*}

\textbf{Static Partitioning Deficiencies.} The majority of existing frameworks, such as split federated learning approaches \cite{Thapa2022, Sun2024split}, determine the optimal split point \textit{offline} or at initialization. They assume a stable resource profile.
However, edge environments are defined by \textbf{volatility}. A device's available CPU cycles fluctuate rapidly due to background OS tasks or thermal throttling \cite{Sakip2023}.
A static split point that is optimal at $t=0$ often becomes a bottleneck at $t=1$. For instance, if the network degrades, a server-heavy split causes stalls. If the CPU load spikes, an edge-heavy split causes system-wide latency.

\textbf{The Gap in Adaptive Control (C2).} Recent adaptive approaches (e.g., Rule-Based Split in Table \ref{tab:gap_analysis}) rely on reactive bandwidth thresholds \cite{Huang2023roof}, suffering from two flaws: (1) \textit{Lack of Lookahead}—reacting only after degradation, causing dropped frames; (2) \textit{Optimization Blindness}—treating all inputs equally despite varying difficulty. Table \ref{tab:gap_analysis} demonstrates no existing paradigm satisfies continuous edge learning requirements, necessitating StreamSplit's co-designed \textit{Uncertainty-Guided} framework.
\section{Our Quality Metrics: Affinity and Diversity}
\label{sec:metrics}

To systematically address the trade-off between representation quality and edge constraints, we propose tractable metrics that quantify the two essential attributes of a robust embedding space: \textbf{Diversity} and \textbf{Affinity}.
We formalize these attributes using the geometry of the induced embedding distribution $p_\theta(z)$.
Our quantification process is illustrated in Figure \ref{fig:metrics_concept}. In the following subsections, we provide a detailed explanation of the theoretical framework, along with empirical evidence demonstrating the advantages of our distribution-based metrics.

\subsection{Diversity: Ensuring Global Separation}
\label{sec:metrics:diversity}

The "Small-Batch Conflict" (C1) primarily threatens the global structure of the latent space. Standard contrastive losses require large batches to approximate the partition function; without them, the model minimizes loss by mapping all inputs to a single point $c$, a phenomenon known as \textit{dimensional collapse} \cite{Wang2020}.

\textbf{Definition 1 (Diversity).} We define Diversity as the alignment between the marginal distribution of embeddings $p_\theta(z)$ and a fixed, high-entropy prior distribution $q(z)$ (typically uniform on the hypersphere $\mathcal{U}(\mathbb{S}^{d-1})$).
\begin{equation}
    \mathcal{D}_{div} = \mathcal{W}_2(p_\theta(z), \mathcal{U}(\mathbb{S}^{d-1}))
\end{equation}
where $\mathcal{W}_2$ denotes the Wasserstein distance. High diversity (minimizing $\mathcal{D}_{div}$) implies that the entropy $H(p_\theta(z))$ is maximized.

\textbf{Definition Explanation.} High diversity (minimizing $\mathcal{D}_{div}$) ensures embeddings distribute uniformly across the hypersphere, maximizing discriminative capacity even with small local batches (Figure \ref{fig:metrics_concept}a).

\begin{theorem}[Small-Batch Robustness]
\label{thm:diversity}
Let $p_\theta$ denote the edge embedding distribution and $\mathcal{U}$ the uniform distribution on $\mathbb{S}^{d-1}$. If $p_\theta$ has $\epsilon$-diversity (i.e., $\mathcal{W}_1(p_\theta, \mathcal{U}) < \epsilon$), then for batch size $N$, the generalization gap between the empirical contrastive loss $\mathcal{L}_N$ and the population loss $\mathcal{L}_\infty$ is bounded by:
\begin{equation}
    |\mathcal{L}_{N} - \mathcal{L}_{\infty}| \le C_1 \epsilon + \frac{C_2}{\sqrt{N}}
\end{equation}
where $C_1, C_2 > 0$ are constants depending on the critic function's Lipschitz 
constant.
\end{theorem}
This theorem implies that minimizing the distributional gap $\epsilon$ directly compensates for the bias induced by small batch sizes $N$. (See Appendix A for proof).

\textbf{Metric: Sliced-Wasserstein Distance (SWD).} A potential concern is that computing the full Wasserstein distance is computationally expensive ($O(d^3)$). To ensure tractability on edge devices, we employ the Sliced-Wasserstein Distance (SWD) as a scalable proxy ($O(Md \log d)$). By projecting the high-dimensional distributions onto a set of $M$ random unit vectors $\Omega = \{\omega_m\}_{m=1}^M$, we compute the closed-form solution:
\begin{equation}
    \label{eq:swd}
    \mathcal{L}_{SW} = \frac{1}{M} \sum_{m=1}^{M} \int_0^1 \left| F^{-1}_{\theta, \omega_m}(\tau) - F^{-1}_{q, \omega_m}(\tau) \right|^2 d\tau
\end{equation}
where $F^{-1}$ is the inverse CDF. Minimizing $\mathcal{L}_{SW}$ forces the sorted projections of edge embeddings to match the uniform distribution, ensuring global coverage.

\vspace{0.5em}
\noindent\textbf{Remark (Metric Equivalence).} While Theorem~\ref{thm:diversity} relies on the Wasserstein-1 distance $\mathcal{W}_1$, our optimization minimizes the Sliced-Wasserstein distance $\mathcal{L}_{SW}$. On the compact hypersphere $\mathbb{S}^{d-1}$, it is established that $\mathcal{L}_{SW}$ is topologically equivalent to $\mathcal{W}_p$ distances \cite{Kolouri2019}. Specifically, convergence in $\mathcal{L}_{SW}$ implies convergence in $\mathcal{W}_1$, making Eq.~\ref{eq:swd} a tractable proxy for minimizing the theoretical error bound $\epsilon$.

\subsection{Affinity: Preserving Local Structure}
\label{sec:metrics:affinity}

While Diversity ensures global space utilization, Affinity ensures the mapping 
preserves the temporal topology of the input signal. In the "Volatility Conflict" 
(C2), sparse updates can lead to a "jagged" manifold where temporally adjacent 
frames $x_t, x_{t+1}$ map to distant points $z_t, z_{t+1}$.

\textbf{Motivation.} For a Lipschitz-continuous encoder, small temporal shifts in the input should produce small shifts in the embedding space. Formally, if $||x_t - x_{t+1}|| \le \delta$, then a well-behaved encoder satisfies $||f_\theta(x_t) - f_\theta(x_{t+1})|| \le K\delta$ for some Lipschitz constant $K$. This motivates measuring smoothness directly in embedding space.

\textbf{Definition 2 (Affinity).} Let $\mathcal{G} = (V, E, W)$ be a temporal 
adjacency graph where vertices $V = \{z_1, \ldots, z_T\}$ are embeddings and 
edges $E$ connect temporally adjacent frames with weights $W_{ij}$. We define 
Affinity as the inverse of the \textit{Dirichlet energy} (graph Laplacian 
quadratic form):
\begin{equation}
    \mathcal{D}_{aff} = \frac{1}{|E|} \sum_{(i,j) \in E} W_{ij} ||z_i - z_j||^2 
    = \frac{1}{|E|} \text{Tr}(\mathbf{Z}^\top \mathbf{L} \mathbf{Z})
\end{equation}
where $\mathbf{L} = \mathbf{D} - \mathbf{W}$ is the graph Laplacian and 
$\mathbf{D}$ is the degree matrix. High affinity corresponds to \textit{low} 
$\mathcal{D}_{aff}$, indicating smooth embeddings over the temporal graph.

\textbf{Definition Explanation (Acoustic Consistency).} Unlike video (which has scene cuts) or images (which are independent), ambient audio is physically continuous. Sound sources do not teleport in feature space. Minimizing $\mathcal{D}_{aff}$ enforces this \textbf{physical inertia}, ensuring the manifold remains smooth even when network dropouts cause the server to receive sparse updates.

\begin{theorem}[Temporal Interpolation]
\label{thm:affinity}
Let $\mathcal{G}$ be a connected temporal graph with spectral gap $\lambda_2 > 0$ 
(second smallest eigenvalue of $\mathbf{L}$). If the embedding sequence has 
Dirichlet energy $\mathcal{D}_{aff} \le \alpha$, then the mean squared error 
of reconstructing any missing frame $z_{t^*}$ via weighted neighbor averaging is bounded:
\begin{equation}
    \mathbb{E}\left[ ||z_{t^*} - \hat{z}_{t^*}||^2 \right] \le \frac{2\alpha \cdot |E|}{\lambda_2 \cdot |\mathcal{N}(t^*)|}
\end{equation}
where $\hat{z}_{t^*} = \frac{1}{|\mathcal{N}(t^*)|}\sum_{j \in \mathcal{N}(t^*)} z_j$ 
is the neighbor average and $\mathcal{N}(t^*)$ denotes the temporal neighbors of $t^*$.
\end{theorem}
This implies that a high-affinity representation (low $\alpha$) provides intrinsic 
robustness to the volatility of edge execution. (See Appendix A for proof).

\textbf{Metric: Laplacian Regularization.} We directly optimize the Dirichlet 
energy as the affinity loss:
\begin{equation}
    \mathcal{L}_{Lap} = \frac{1}{|E|} \sum_{(i,j) \in E} W_{ij} ||z_i - z_j||^2
\end{equation}
Minimizing $\mathcal{L}_{Lap}$ explicitly penalizes "jagged" transitions, 
smoothing the manifold to maintain semantic coherence.
\subsection{Metric Advantage}
\label{sec:metrics:validation}

Given the fine granularity of our metrics, we empirically investigate their correlation with task performance. We conduct controlled experiments on AudioSet-Balanced (20,371 samples, 527 classes) by systematically varying affinity and diversity.

\textbf{Diversity Validation.} We degrade diversity by restricting the uniform prior to spherical cones of half-angle $\theta$ ranging from 10° to 90° in 10° increments. At $\theta = 10°$, embeddings are forced into a narrow cone (severe collapse); at $\theta = 90°$, the prior covers the full hypersphere. The results reveal a strong negative correlation between $\mathcal{L}_{SW}$ and downstream accuracy (Pearson $r = -0.96$, $p < 0.001$), significantly outperforming Maximum Mean Discrepancy (MMD, $r = 0.82$, $p < 0.01$) \cite{Gretton2012}. At $\theta = 10°$ ($\mathcal{L}_{SW} = 0.89$), accuracy drops to 55.2\%; at $\theta = 90°$ ($\mathcal{L}_{SW} = 0.08$), accuracy reaches 73.1\%.

\textbf{Affinity Validation.} We inject temporal discontinuities by randomly shuffling frames within a 3-second window with probability $p$ ranging from 0.0 to 0.8 in 0.1 increments. The $\mathcal{L}_{Lap}$ metric shows strong positive correlation with accuracy degradation (Pearson $r = 0.93$, $p < 0.001$). At $p = 0.8$ (severe jitter, $\mathcal{L}_{Lap} = 2.31$), accuracy degrades by 12.4\%; at $p = 0$ (no jitter, $\mathcal{L}_{Lap} = 0.35$), accuracy is 72.8\%. The spectral gap $\lambda_2$ decreases from 0.42 to 0.08 as jitter increases, confirming that temporal disruptions degrade manifold connectivity as predicted by Theorem \ref{thm:affinity}.

These results demonstrate that our distribution-based metrics are more sensitive to quality variations than existing approaches, making them reliable objectives for the StreamSplit optimization loop (Section \ref{sec:framework}).
\section{StreamSplit Framework}
\label{sec:framework}

StreamSplit transforms the traditionally static edge-cloud link into a dynamic, 
feedback-driven control loop. Unlike "train-then-compress" paradigms which fix 
model architecture at deployment time, StreamSplit treats the neural network as 
a flexible pipeline that can be partitioned, compressed, and scheduled in real-time.

As illustrated in Figure \ref{fig:architecture}, the framework is architected as an asymmetric distributed system spanning edge devices and the cloud. The figure explicitly maps out the data and control flow, detailing the key components and interaction pathways within three decoupled subsystems: the \textbf{Edge Learner} (Phase 1), which maximizes local throughput on the edge devices under memory constraints; the \textbf{Control Plane} (Phase 2), which governs adaptive offloading based on system state and data difficulty; and the \textbf{Cloud Refiner} (Phase 3), which enforces global consistency across the manifold despite asynchronous updates.

\subsection{Phase 1: Edge-Side Streaming Execution}

The primary operational constraint on the edge is memory volatility. Standard 
Contrastive Learning (CL) requires storing thousands of negative samples (e.g., 
in a Memory Bank or Large Batch) to approximate the global data distribution. 
On memory-constrained devices, maintaining large queues causes cache thrashing 
when background tasks compete for resources \cite{Fang2022}. This raises a critical 
design question:

\begin{rqbox}
\textbf{RQ1:} How can we maintain high-quality contrastive representations on 
memory-constrained edge devices without requiring large negative sample batches?
\end{rqbox}

\subsubsection{Analysis Method}
We analyze the memory-quality trade-off throughout the contrastive learning 
process. Standard approaches store $N$ negative samples requiring $O(Nd)$ memory. 
For typical settings ($N=4096$, $d=128$), this consumes 512KB of RAM, causing 
cache thrashing on edge devices. We investigate whether a generative distribution 
can replace explicit sample storage while maintaining representation quality.

\subsubsection{Distributional Memory Footprint}
To resolve the ``Small-Batch Conflict'' (C1), we replace memory-intensive sample queues with a generative distribution. Instead of storing discrete raw tensors, the Edge Learner maintains a lightweight \textbf{Gaussian Mixture Model (GMM)}:
\begin{equation}
    p_{local}(z) = \sum_{c=1}^{C} \pi_c \mathcal{N}(z; \mu_c, \Sigma_c)
    \label{eq:gmm}
\end{equation}
where $\pi_c$ are mixing weights, $\mu_c \in \mathbb{R}^d$ are component means, and $\Sigma_c$ are diagonal covariance matrices. With $C=64$ components and embedding dimension $d=128$, utilizing FP16 precision (2 Bytes), the storage requirement is:
\begin{equation}
    \text{Size} \approx 2 \cdot (C \times d \times 2\text{B}) + (C \times 2\text{B}) \approx 33\text{ KB}
\end{equation}
While larger than our initial estimate, this fits comfortably within the L1/L2 cache of standard ARM processors (e.g., Cortex-A72 has 1MB L2), whereas a standard contrastive memory bank (e.g., 16k samples) requires $>8$MB of RAM and incurs high DRAM access energy.

This GMM is updated incrementally using a streaming Expectation-Maximization (EM) algorithm \cite{Ghahramani2000} with $O(Cd)$ complexity per iteration. To ensure stability during initialization (Cold Start), the system defaults to a conservative local policy for the first 50 frames to populate the sufficient statistics.

\subsubsection{Hard Negative Mining via GMM}
Merely having a distribution is insufficient; we must sample from it effectively. 
Standard small-batch training fails because randomly sampled negatives are often 
"easy" (semantically distant), providing negligible gradients \cite{Chen2020}.

We implement \textbf{Boundary-Aware Sampling}, visualized in Figure 
\ref{fig:das_logic}. During the forward pass, given an anchor embedding $z^+$ 
from component $c^*$, we sample "virtual" negative embeddings $z^-$ from 
\textit{other} components weighted by proximity to the anchor:
\begin{equation}
    p(c|z^+, c^*) \propto \pi_c \cdot \exp\left(-\frac{||\mu_{c^*} - \mu_c||^2}{2\tau^2}\right), 
    \quad c \neq c^*
    \label{eq:hard_sampling}
\end{equation}
where $\tau$ is a temperature parameter controlling hardness. This strategy 
samples negatives near the decision boundary between components (hard negatives) 
rather than from distant components (easy negatives) or out-of-distribution 
regions. By synthesizing virtual negatives on-the-fly, we effectively augment the batch size without increasing physical memory usage, as virtual samples are generated, $\ell_2$-normalized to project them onto the hypersphere, and immediately discarded after gradient computation.


\begin{figure}[t]
\centering
\includegraphics[width=0.9\columnwidth]{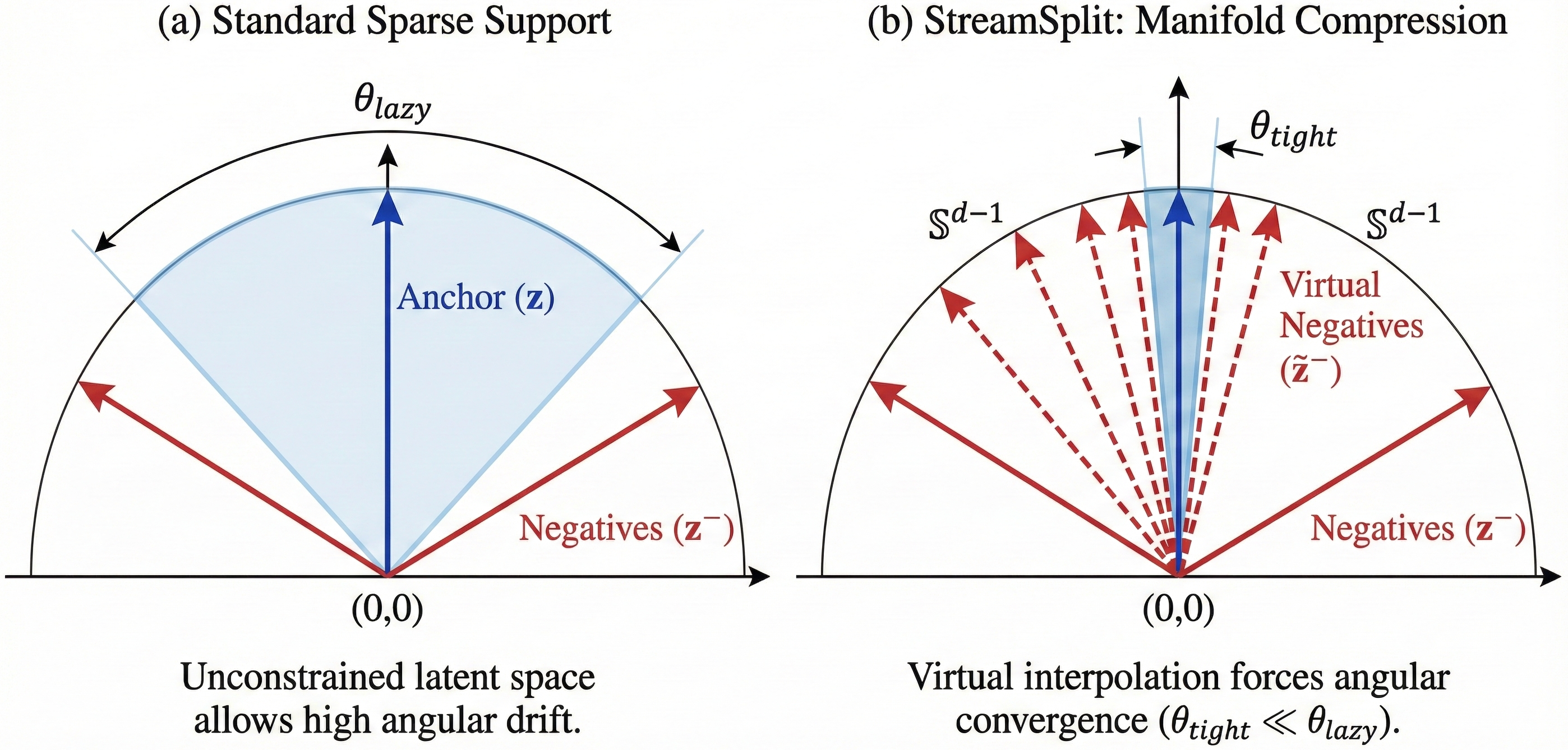}
\caption{\textbf{Addressing Small-Batch Conflict.} (a) Small batches cause ``Lazy Margins'' ($\theta_{lazy}$) and degradation. (b) StreamSplit generates virtual negatives from a compact GMM (<35KB), enforcing ``Tight Margins'' ($\theta_{tight}$) to decouple quality from physical batch size.}
\label{fig:das_logic}
\end{figure}

\textbf{Edge Training Objective and Flow.} 
To strictly minimize on-device memory, we employ a streaming InfoNCE loss with virtual negatives. For each incoming audio frame $x_t$, we generate a positive pair $(\tilde{x}_t, \tilde{x}'_t)$ via standard lightweight augmentations (random Gaussian noise and frequency masking). We do \textit{not} use temporal neighbors as positives to avoid buffering latency.

The edge encoder $f_\theta$ minimizes the negative log-likelihood of the positive pair against $N_{syn}$ virtual negatives sampled from the GMM:
\begin{equation}
    \mathcal{L}_{edge} = -\log \frac{\exp(\text{sim}(z_t, z'_t) / \tau)}{\exp(\text{sim}(z_t, z'_t) / \tau) + \sum_{j=1}^{N_{syn}} \exp(\text{sim}(z_t, z_{syn}^j) / \tau)}
\end{equation}
where $z_t = f_\theta(\tilde{x}_t)$, $z'_t = f_\theta(\tilde{x}'_t)$, and $\{z_{syn}^j\}_{j=1}^{N_{syn}}$ are virtual negatives sampled from the local GMM $p_{local}(z)$ using the boundary-aware strategy (Eq.~\ref{eq:hard_sampling}). 

\textbf{Training Flow.} This objective allows for a fully streaming update:
\begin{enumerate}
    \item \textbf{Forward:} Encoder processes $x_t$ to get $z_t$; GMM is updated via online EM.
    \item \textbf{Sample:} Virtual negatives $z_{syn}$ are generated from GMM parameters.
    \item \textbf{Backward:} Gradients from $\mathcal{L}_{edge}$ update the encoder weights $f_\theta$.
    \item \textbf{Sync:} Updated weights are lazily synchronized with the server (See Sec.~\ref{sec:framework}-Phase 3) to align with the global manifold.
\end{enumerate}

\begin{answerbox}
\textbf{Answer to RQ1:} Distributional memory replaces explicit sample storage 
with a compact GMM ($O(Cd)$ vs. $O(Nd)$ memory), while boundary-aware sampling 
synthesizes hard negatives on-the-fly. This decouples representation quality from 
physical batch size, enabling edge devices to maintain contrastive learning 
performance despite memory constraints.
\end{answerbox}

\subsection{Phase 2: Uncertainty-Guided Control}

The core systems innovation of StreamSplit is the \textbf{Control Plane}, which 
resolves the "Volatility Conflict" (C2). Unlike heuristic splitters (e.g., 
Neurosurgeon \cite{Kang2017}) that react only to bandwidth changes, our system 
must proactively manage workload based on both system state and data difficulty. 
This motivates our second research question:

\begin{rqbox}
\textbf{RQ2:} How can we dynamically adapt the edge-cloud computation partition 
to handle runtime volatility in CPU load, network bandwidth, and data difficulty?
\end{rqbox}

\subsubsection{Analysis Method}
We analyze the runtime adaptation problem across the entire execution pipeline, 
including edge compute constraints, network transmission costs, and varying data 
complexity. Existing heuristic methods fail when multiple constraints conflict 
(e.g., high CPU load but low bandwidth). We investigate whether a learned policy 
can discover non-obvious adaptation strategies that balance competing objectives.

\subsubsection{RL Formulation}
\label{sec:rlformulatiom}
We formulate the runtime adaptation as a Markov Decision Process (MDP) defined by the tuple $(\mathcal{S}, \mathcal{A}, \mathcal{R})$, with formal definitions provided in Appendix \ref{app:controlplane}. The Control Plane runs as a lightweight sidecar process on the edge device, executing policy inference at regular intervals to avoid excessive overhead.

\textbf{State Space ($s_t$):} The state vector $s_t = [U_t, R_{cpu}, 
    B_{net}]$ captures the full system context. $R_{cpu} \in [0,100]$ represents 
    CPU utilization percentage. $B_{net} \in \mathbb{R}^+$ is the estimated 
    uplink bandwidth, computed via exponential moving average over recent 
    transmission times. Crucially, we include $U_t \in [0, \log C]$, the 
    \textbf{embedding uncertainty}, defined as the Shannon entropy of the GMM 
    component assignment:
    \begin{equation}
        U_t = H(p(c|z_t)) = -\sum_{c=1}^{C} p(c|z_t) \log p(c|z_t)
        \label{eq:uncertainty}
    \end{equation}
    where $p(c|z_t) = \frac{\pi_c \mathcal{N}(z_t; \mu_c, \Sigma_c)}{\sum_{c'} 
    \pi_{c'} \mathcal{N}(z_t; \mu_{c'}, \Sigma_{c'})}$ is the posterior 
    probability via Bayes' rule. High entropy typically correlates with \textbf{transient acoustic events} (e.g., speech onset, glass breaking) which are semantically complex. Low entropy corresponds to steady-state background noise (e.g., HVAC hum). By conditioning the split on $U_t$, the agent effectively learns a semantic Event Detection policy, offloading only the information-dense frames. We prefer GMM entropy over TTA variance (high latency) or kNN density (high memory) as it provides a zero-cost uncertainty proxy computed directly within the standard inference pass.

\textbf{Action Space ($a_t$):} The action is the split layer index 
    $k \in \{0, 1, \dots, L\}$, where $L$ is the total number of encoder layers. 
    Selecting $k=0$ implies full offloading (minimal edge compute, maximum 
    bandwidth consumption), while $k=L$ implies full on-device processing 
    (maximum edge compute, zero bandwidth). Note that to satisfy strict bandwidth constraints detected in the state $s_t$, specific offloading actions ($k < L$) are coupled with dynamic quantization (e.g., INT8 precision) to ensure the transmission fits within the estimated $T_{max}$.

\textbf{Reward Function ($r_t$):} The reward balances task performance 
    with system efficiency:
    \begin{equation}
        r_t = \alpha \cdot \mathcal{A}_{task} - \beta \cdot \frac{\text{Lat}_t}{T_{max}} 
              - \eta \cdot \frac{E_t}{E_{budget}}
        \label{eq:reward}
    \end{equation}
    where $\mathcal{A}_{task} \in [0,1]$ is the task accuracy on a validation set, $\text{Lat}_t$ is the end-to-end latency (edge compute + network transmission + server compute) normalized by a maximum allowable latency $T_{max}$, and $E_t$ is the energy consumption normalized by a per-frame budget $E_{budget}$. The coefficients $\alpha$, $\beta$, $\eta$ prioritize accuracy while penalizing latency violations and energy drain.


\textbf{Decision Frequency and Atomic Transitions:} The RL agent executes every $T_{step}=10$ frames ($\approx 100$ms). The policy network (2-layer MLP) inference time is measured at $0.2$ms on the Pi 4B. To guarantee consistency during transitions and avoid the complexity of managing in-flight data, split decisions occur atomically at these $100$ms boundaries. When the RL agent shifts the computation partition (e.g., moving from layer 3 to layer 5), the transition applies exclusively to the next $T_{step}$ block. No data currently traversing the pipeline is redone, and no in-flight frames are dropped, stalled, or left in an inconsistent state during the shift. This high-frequency, atomic control allows StreamSplit to react to abrupt bandwidth collapses without corrupting the streaming pipeline.

\subsubsection{Policy Network and Adaptive Behavior}
The control logic is visualized in Figure \ref{fig:rl_flow}. We employ a 
Proximal Policy Optimization (PPO) \cite{Schulman2017} agent with a lightweight 
two-layer fully connected network. The policy network $\pi_\theta(a|s)$ and 
value network $V_\phi(s)$ share the first layer to minimize parameter overhead 
while maintaining expressiveness.

The agent is trained offline on historical traces collected from diverse hardware platforms under varying network conditions (bandwidth, latency, packet loss). The learned policy exhibits complex adaptive behaviors that heuristic rules cannot capture. For instance, under high CPU load but low bandwidth, a heuristic might stall. The StreamSplit agent learns to apply aggressive quantization (reducing embedding precision) to fit bandwidth constraints, sacrificing marginal accuracy to maintain real-time throughput. Conversely, when both CPU and bandwidth are available but uncertainty $U_t$ is high (ambiguous sample), the agent offloads to the server even at minimal split depth, prioritizing accuracy over efficiency. We note that while the training reward $r_t$ uses ground-truth accuracy $\mathcal{A}_{task}$, the deployed policy $\pi_\theta$ relies solely on the observable state $s_t$ (uncertainty, CPU, bandwidth) and does not require labels at runtime.

\begin{figure}[t]
\centering
\includegraphics[width=0.99\columnwidth]{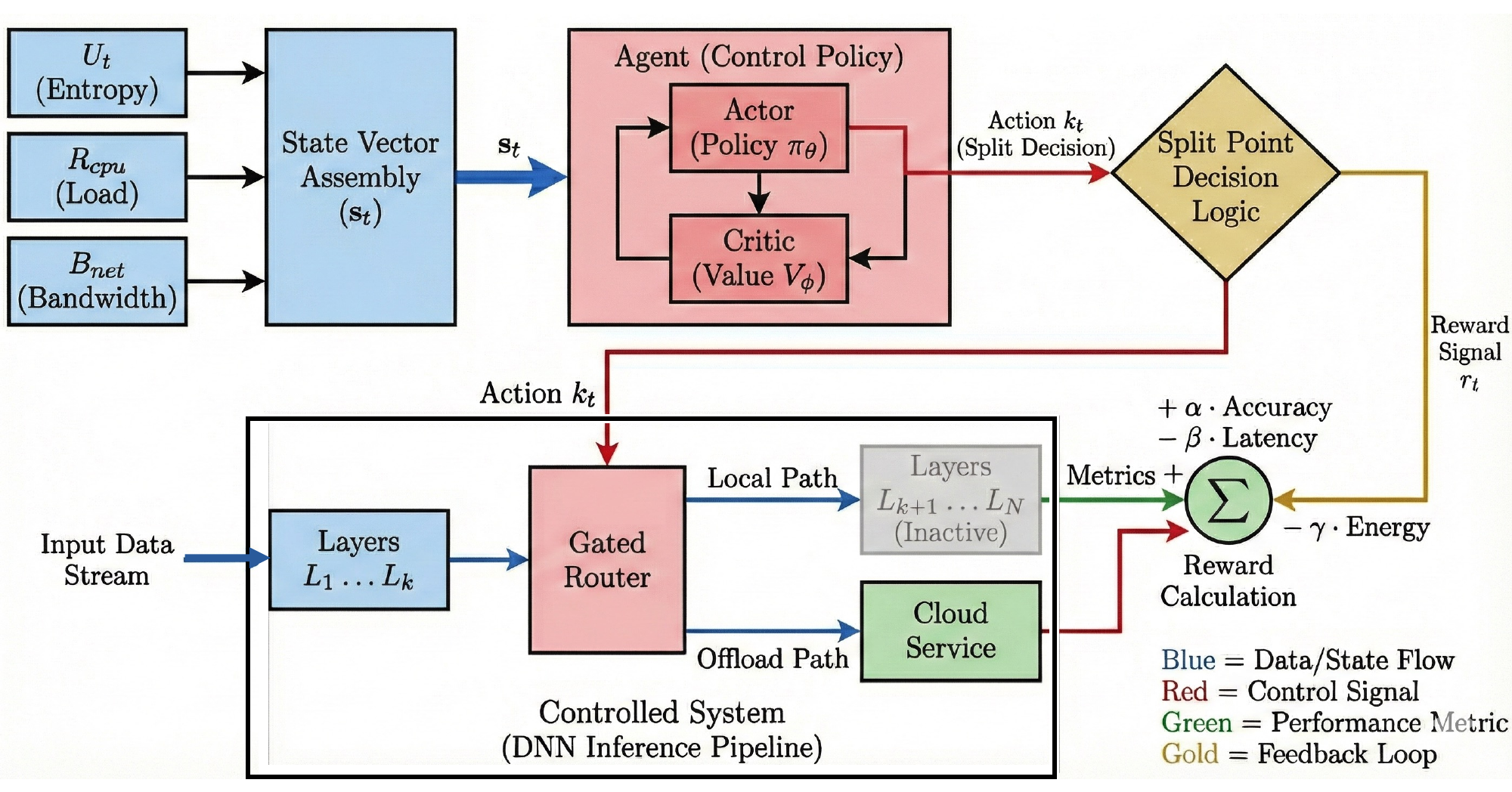}
\caption{\textbf{RL Control Logic.} The agent maps normalized state $\mathbf{s}_t$ (incl. uncertainty) to split layer $k_t$. Reward $r_t$ (Eq. \ref{eq:reward}) balances accuracy, latency, and energy.}
\label{fig:rl_flow}
\end{figure}

\begin{answerbox}
\textbf{Answer to RQ2:} Formulate adaptation as an RL problem where the agent 
observes CPU load, network bandwidth, and embedding uncertainty to select optimal 
split points. The learned policy discovers non-obvious strategies (e.g., 
quantization under bandwidth constraints, uncertainty-based offloading) that 
heuristic rules miss, achieving robust real-time execution across heterogeneous 
platforms.
\end{answerbox}

\subsection{Phase 3: Server-Side Refinement}

The Cloud Server handles the heavy lifting of global model convergence. However, 
dynamic splitting introduces \textbf{asynchrony}: frames arrive out of order or 
are dropped entirely when the network degrades or the RL agent selects full 
on-device processing to conserve bandwidth. This breaks standard synchronous 
SGD, which assumes uniform batches and temporal ordering. To address this, the 
Server Refiner implements a \textbf{Temporal Buffer} and the \textbf{Hybrid Loss} 
derived in Section \ref{sec:metrics}.

\subsubsection{Temporal Buffer Management}
As shown in Figure \ref{fig:server_buffer}, the server maintains a sliding 
temporal window of recent embeddings. When an update $z_t$ arrives from the 
edge, it is inserted into the buffer at the corresponding temporal index. If a 
frame is skipped (due to RL-driven local processing or network packet loss), 
the buffer contains a temporal gap. Rather than discarding the buffer or 
performing explicit interpolation (which would require maintaining expensive 
historical state), we construct a k-Nearest Neighbor temporal graph 
$\mathcal{G} = (V, E)$ on the available embeddings, where edges $(i,j) \in E$ 
connect temporally adjacent frames within a small window. We set the server temporal window size to $W=100$ frames (approx. 1s context). The $k$-NN graph is constructed with $k=5$ neighbors. Constructing the Laplacian on the server takes approx. $3$ms for a batch of 100, which is fully pipelined with inference to minimize overhead.

\subsubsection{Hybrid Objective for Robust Convergence}
The server optimizes the combined objective from Eq. \ref{eq:server_loss}:
\begin{equation}
    \mathcal{L}_{server} = \underbrace{\mathcal{L}_{task}}_{\text{Classification}} 
    + \lambda_1 \underbrace{\mathcal{L}_{SW}(p_\theta, \mathcal{U})}_{\text{Diversity}} 
    + \lambda_2 \underbrace{\mathcal{L}_{Lap}(\mathcal{G})}_{\text{Affinity}}
    \label{eq:server_loss}
\end{equation}
where $\mathcal{L}_{task}$ denotes the global objective. In pure self-supervised settings, $\mathcal{L}_{task}$ is formulated as a global InfoNCE loss over the server buffer; when sparse labels are available (as in our validation experiments to establish performance upper bounds), it is instantiated as the standard cross-entropy loss. $\mathcal{L}_{SW}$ is the Sliced-Wasserstein Distance (Eq. \ref{eq:swd}) enforcing global diversity, and $\mathcal{L}_{Lap}$ is the Laplacian regularization (from Section \ref{sec:metrics}) enforcing temporal smoothness:
\begin{equation}
    \mathcal{L}_{Lap} = \frac{1}{|E|} \sum_{(i,j) \in E} ||z_i - z_j||^2
\end{equation}


The Laplacian term penalizes large embedding jumps between temporally adjacent frames. This mechanism is critical for handling extended periods of connectivity loss. During a severe network outage, the edge device may be forced to drop frames if local memory is exhausted. When connectivity resumes, the server's Temporal Buffer receives a sequence with a massive temporal gap. Instead of naively interpolating missing features---which would hallucinate phantom acoustic events---$\mathcal{L}_{Lap}$ enforces a Lipschitz-continuous smoothing prior across the boundary of the outage. 

As formalized by Theorem \ref{thm:affinity}, minimizing the Dirichlet energy guarantees bounded interpolation error even when the spectral gap ($\lambda_2$) of the temporal graph is degraded by dropouts. The server effectively "stitches" the representations together by pulling the pre-outage and post-outage embeddings toward a shared, smooth manifold geometry. Because the representations remain Lipschitz-continuous, the server can bridge these temporal gaps gracefully. This is empirically validated in our ablation studies (Table \ref{tab:ablation_loss}), where StreamSplit sustains a severe 40\% frame drop rate while only degrading 6.6\% in accuracy, whereas standard contrastive objectives suffer catastrophic collapse.

Simultaneously, the Sliced-Wasserstein term prevents mode collapse by ensuring 
the aggregated buffer distribution matches the global uniform prior. This is 
particularly important in asynchronous federated settings, where devices with 
fast connections might dominate the global model. By minimizing $\mathcal{L}_{SW}$, 
we ensure that the server model converges to a balanced representation of the 
entire data distribution, rather than overfitting to high-throughput devices.

\begin{figure}[t]
\centering
\includegraphics[width=0.76\columnwidth]{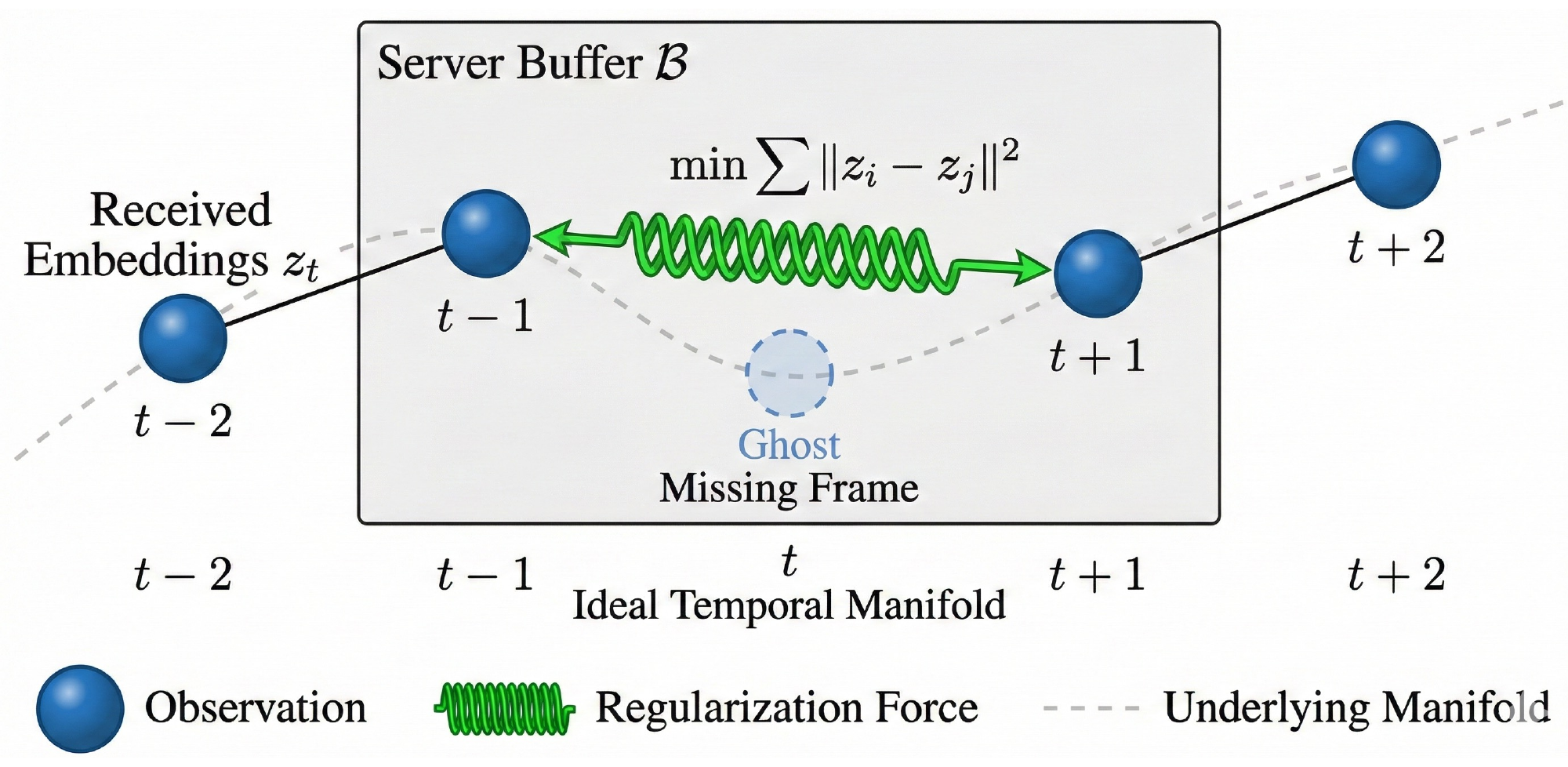}
\caption{\textbf{Manifold Stitching.} Laplacian regularization ($\mathcal{L}_{Lap}$) acts as a spring force across temporal gaps (missing frames), maintaining manifold continuity without requiring explicit interpolation.}
\label{fig:server_buffer}
\end{figure}

\subsubsection{Downlink Synchronization Strategy}
To close the learning loop without incurring prohibitive downlink costs, we employ a \textbf{Lazy Synchronization} protocol. The Server transmits updated GMM parameters ($\mu_c, \Sigma_c$)—which are lightweight ($<35$KB)—to the Edge every $T_{sync}=100$ frames. The heavier Encoder weights ($f_\theta$) are only synchronized when the device detects a charging state or high-bandwidth WiFi connection. In our energy evaluation (Table \ref{tab:energy}), we account for the GMM synchronization cost, which adds a negligible $0.4$ mJ/frame average overhead.

\section{Implementation}
\label{sec:implementation}

\textbf{Definitions.} We define a \textbf{Sample} as a 1-second audio clip. A \textbf{Batch} consists of $N=8$ samples. A \textbf{Frame} refers to the STFT window (25ms). The RL agent makes decisions every $T_{step}=10$ frames ($\approx 100$ms) to ensure responsiveness to network volatility, while gradients are aggregated at the \textbf{Batch} level ($N=8$).

We implement the StreamSplit framework entirely in Python, using PyTorch~\cite{Paszke2019} 
for model training and inference across all components. Edge deployment leverages 
PyTorch's TorchScript for optimized execution on resource-constrained devices.

\textbf{Audio Processing.} For real-time spectrogram extraction, we employ 
PyKissFFT~\cite{KissFFT}, a Python binding for the lightweight KissFFT library 
optimized for embedded platforms. Compared to NumPy's FFT implementation, 
PyKissFFT reduces STFT computation latency from 7.8ms to 3.2ms per frame on 
Raspberry Pi 4B, enabling continuous 16kHz audio streaming without frame drops. 
We compute 128-bin mel spectrograms using 25ms windows with 10ms hop length.

\textbf{Hardware Platforms.} We deploy and evaluate StreamSplit on three 
representative platforms spanning the edge-cloud spectrum:
\begin{itemize}
    \item \textbf{Raspberry Pi 4B} (4GB RAM, ARM Cortex-A72 @ 1.5GHz): 
    Represents resource-constrained IoT devices. The GMM module (35KB) and PPO 
    policy (12KB) fit entirely within the 256KB L2 cache.
    \item \textbf{Apple MacBook M2} (8GB unified memory, 8-core CPU): Represents capable edge devices with integrated neural engine. Used for development and mid-tier deployment scenarios. We leverage the PyTorch \texttt{MPS} (Metal Performance Shaders) backend to offload matrix operations to the M2's 10-core \textbf{GPU}, achieving a $3.2\times$ speedup over CPU execution. While the Neural Engine (ANE) offers further theoretical gains, current framework support favors the GPU for custom dynamic graphs like ours.
    \item \textbf{Cloud Server} (Intel Xeon Gold 6248R, 4$\times$ NVIDIA RTX 
    3090): Hosts the Server Refiner for global model aggregation and refinement.
\end{itemize}

\textbf{Network Emulation.} To evaluate adaptation under realistic network 
conditions, we employ the Linux Traffic Control (\texttt{tc}) utility as a 
network emulator, simulating bandwidth fluctuations (1--50 Mbps), latency 
variations (20--200ms), and packet loss (0--5\%). We construct 6 network 
profiles based on real-world 4G/5G traces~\cite{Yan2018}, spanning stable, 
variable, and congested scenarios.



\textbf{Training Configuration.} Edge-side training uses the Adam optimizer with a learning rate of $10^{-3}$ and batch size $N=8$. The PPO agent trains for 2M environment steps with a discount factor $\gamma=0.99$ and clipping $\epsilon=0.2$. Server-side training employs SGD with momentum 0.9 over 100 epochs. To find a balance between accuracy and strict penalties for latency and energy use, we use a lightweight 2-layer MLP for the RL policy with specific reward weights: $\alpha=10$, $\beta=5$, and $\eta=3$. Loss hyperparameters ($\lambda_1=0.1$, $\lambda_2=0.01$) were determined via grid search on held-out validation data.

\textbf{Reproducibility Details.} The audio encoder backbone is a standard ResNet-18 adapted for 1D audio (taking spectrogram inputs), consisting of $L=8$ splitable blocks with an output dimension $d=128$, which provides the standard latent size for audio representation stability. For the Distributional Memory, we utilize a GMM with $C=64$ components ($\tau=0.1$) to find the optimal balance between expressiveness and cache fit. Crucially, this estimator only takes up $33$ KB of space, allowing it to fit completely in the L2 cache of standard ARM processors. This avoids having to access DRAM, which uses a lot of power. During hard negative mining, we sample $N_{syn}=256$ virtual negatives per anchor to ensure a sufficient gradient signal without physical memory overhead. 

\textbf{Execution Consistency.} To address in-flight data and ensure consistency during dynamic transitions, our framework makes split decisions atomically at $100$ms boundaries ($T_{step}=10$ frames) to ensure real-time response to network volatility. This means that split-point shifts happen precisely between frames, ensuring that no data is redone and there is no lag in flight.

\textbf{SWD and Sampling Overhead.} We set the number of random projections for the Sliced-Wasserstein Distance to $M=50$ on both the edge and server to balance approximation error with compute cost. We measured the specific runtime overhead on the Raspberry Pi 4B: the SWD computation introduces a latency of $1.2$ms per batch, while the GMM-based hard negative synthesis (generating $N_{syn}=256$ virtual negatives) adds only $0.8$ms per batch. These overheads are negligible compared to the $10$ms hop length, preserving real-time capabilities.

\textbf{Quantization Implementation.} To minimize transmission latency during offloading ($k < L$), we compress intermediate feature maps using \textbf{asymmetric INT8 quantization} with a per-tensor granularity. We employ Post-Training Quantization (PTQ) rather than Quantization-Aware Training (QAT) to maintain training pipeline simplicity. Calibration statistics (min/max ranges) are collected offline using a representative subset of the training data. The runtime overhead for quantization and dequantization on the Raspberry Pi 4B is measured at $<0.5$ms per frame. Empirical evaluation shows this compression introduces a negligible accuracy degradation of $<0.3\%$ compared to full FP32 transmission, validating its viability for bandwidth-constrained edges.

\section{Evaluation}
\label{sec:eval}
We conduct extensive experiments to evaluate StreamSplit across diverse datasets, hardware platforms, and network conditions through systematic evaluation of system efficiency, representation quality, component contributions, and hardware generalization.

\subsection{Experimental Setup}
\label{sec:eval:setup}

\subsubsection{Datasets}
We evaluate StreamSplit on two complementary datasets:

\textbf{AudioSet}~\cite{Gemmeke2017} is a large-scale audio classification 
benchmark containing over 2 million human-labeled 10-second clips across 
527 classes. We use the balanced evaluation subset (20,371 clips) for 
standardized comparison with prior work, and a subset of 100K clips for 
training the contrastive encoder.

\textbf{EcoStream-Wild} comprises 48 hours of continuous uncurated audio collected using 4 Raspberry Pi 4B devices deployed in diverse real-world environments: a university office, a crowded cafeteria, a residential living room, and an outdoor balcony, recording over a 2-week period. To establish ground truth, we employed a semi-automated annotation protocol: a silence removal algorithm first proposed candidate segments, which were then manually verified and labeled by human annotators. The dataset features a highly unbalanced, realistic class distribution dominated by ``Silence/Background'' (60.2\%), followed by ``Speech'' (24.5\%), and specific ``Transient Events'' (15.3\%, e.g., footsteps, door slams, glass breaking) across 15 distinct classes. We commit to releasing the full dataset, annotation logs, and source code upon acceptance to facilitate reproducibility.

\subsubsection{Baseline Methods}
We compare StreamSplit against the following representative baselines:

\textbf{Edge-Only:} Full model training and inference on the edge device without server communication, representing maximum autonomy but limited model capacity due to memory constraints.

\textbf{Server-Only:} All audio data transmitted to the cloud server for processing, representing maximum accuracy but high bandwidth consumption and latency.

\textbf{FSL}~\cite{Thapa2022}: Federated Split Learning with fixed partition points, combining federated averaging with split inference.

\textbf{FedCL}~\cite{Zhang2023}: Federated Contrastive Learning using synchronized memory banks across clients with periodic server aggregation.

\textbf{Rule-Based Adaptive Split:} A heuristic baseline that offloads based on static thresholds for \textbf{both} bandwidth and CPU utilization (e.g., ``Offload if Bandwidth > $X$ AND CPU < $Y$''). This represents a competent engineering solution (an upgrade to Neurosurgeon~\cite{Kang2017}) without learning. Note that we exclude early-exit inference methods (e.g., SPINN~\cite{Laskaridis2020}) as baselines, as they optimize for classification confidence rather than the representation quality required for Contrastive Learning.

\subsubsection{Evaluation Metrics}
We evaluate system performance along three dimensions:

\textbf{System Efficiency:} Bandwidth consumption (KB/frame), end-to-end 
latency (ms), and energy consumption (mJ/frame) measured using 
PowerJoular~\cite{PowerJoular} on Raspberry Pi 4B.

\textbf{Representation Quality:} Linear probe accuracy (\%) on frozen 
embeddings, and mean Average Precision (mAP) for audio retrieval tasks.

\textbf{Adaptation:} Response time to network changes (ms) and accuracy 
stability under varying conditions (standard deviation).

\subsection{End-to-End System Performance}
\label{sec:eval:e2e}

We first evaluate StreamSplit's system-level efficiency gains—the critical metrics for practical edge deployment. 

\subsubsection{Bandwidth Reduction}
Figure~\ref{fig:bandwidth} presents bandwidth consumption across methods. Note that the results report the transmitted payload per \textbf{processing batch} of $N=8$ clips (approx. 8 seconds of audio). For the Server-Only baseline, transmitting 8 clips of raw PCM audio (16kHz, 16-bit mono) consumes $8 \times 32\text{KB} = 256\text{KB}$. StreamSplit achieves \textbf{77.1\% bandwidth reduction} (58.7 KB/batch) by transmitting compact intermediate representations.

\textbf{Comparison with Opus Baseline.} We analytically compare against standard Opus compression (VoIP profile at 32kbps). While Opus offers a competitive bandwidth baseline (consuming $\approx 4$ KB/frame, comparable to StreamSplit's most aggressive modes), it introduces a fundamental trade-off that contradicts edge offloading: it forces the cloud server to execute the \textit{entire} feature extraction pipeline ($100\%$ of the compute load). StreamSplit offers a Pareto-optimal frontier: we match Opus bandwidth efficiency when necessary, but typically offload $35\%$ of the compute burden to the edge (Fig.~\ref{fig:latency}) to reduce server latency and load, a benefit that pure audio compression cannot achieve.

Compared to baselines, FSL (187.2 KB/frame) and Rule-Based Split (124.3 KB/frame) achieve only 26.9\% and 51.4\% reduction respectively, as they lack data-aware adaptation. FedCL (201.4 KB/frame) incurs additional overhead from memory bank synchronization.

\begin{figure}[t]
\centering
\includegraphics[width=\columnwidth]{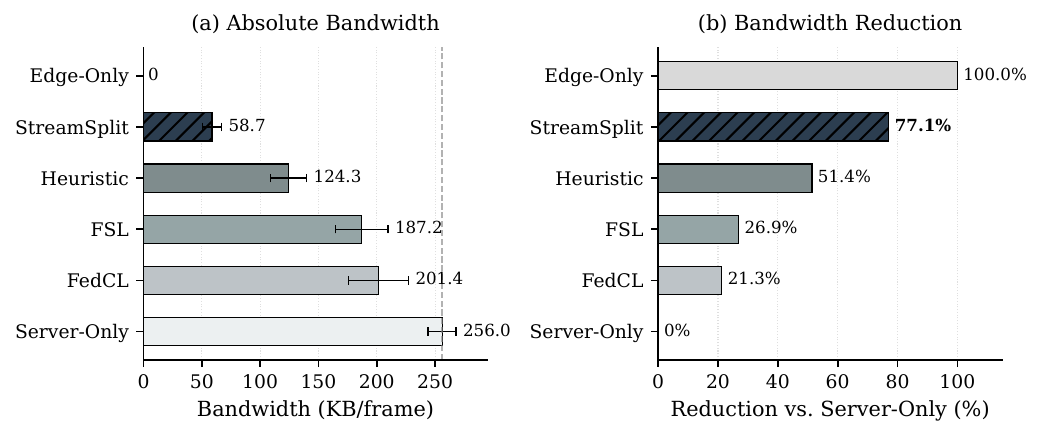}
\caption{\textbf{Bandwidth Consumption.} StreamSplit achieves 77.1\% reduction vs. Server-Only by transmitting compact, adaptive intermediate representations.}
\label{fig:bandwidth}
\end{figure}

\subsubsection{Latency Reduction}
Figure~\ref{fig:latency} shows end-to-end latency breakdown across methods. 
StreamSplit achieves \textbf{72.6\% latency reduction} compared to Server-Only 
(127ms vs. 464ms average). The latency composition reveals the source of gains: StreamSplit spends 45.2ms on edge compute (35.6\%)—which explicitly includes the $<2$ms overhead for GMM updates and RL policy inference—34.6ms on transmission (27.2\%), and 41.2ms on server inference (32.4\%). In contrast, Server-Only spends 312ms on transmission alone (67.2\%) due to raw audio 
upload.

Under degraded network conditions (2 Mbps, 150ms RTT), StreamSplit's 
advantage amplifies: 287ms vs. 1,847ms for Server-Only (84.5\% reduction). 
The RL agent automatically increases local processing, reducing transmission 
dependency.

\begin{figure}[t]
\centering
\includegraphics[width=\columnwidth]{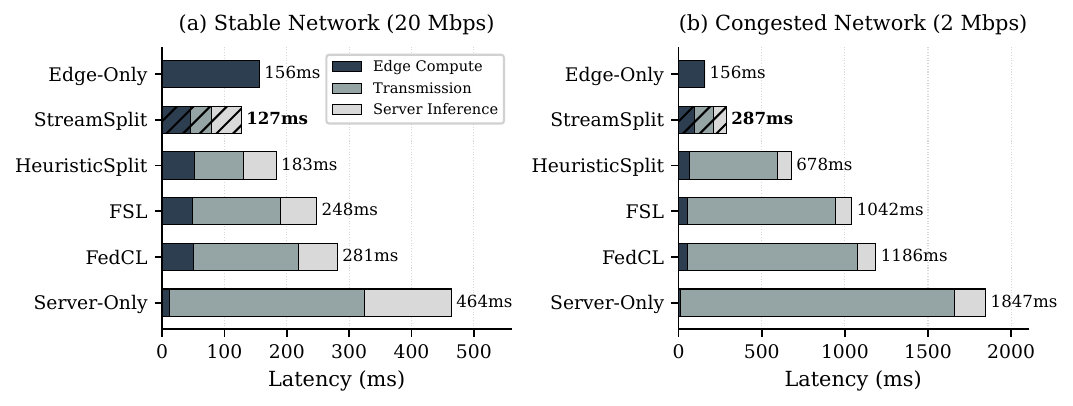}
\caption{\textbf{Latency Breakdown.} StreamSplit reduces latency by 72.6\% (stable) to 84.5\% (congested) via adaptive splitting.}
\label{fig:latency}
\end{figure}

\subsubsection{Energy Savings and Battery Life}
Table~\ref{tab:energy} summarizes energy consumption on Raspberry Pi 4B 
powered by a 10,000mAh battery pack. StreamSplit consumes \textbf{89.3 mJ/frame}, 
achieving \textbf{52.3\% energy savings} compared to Server-Only (187.2 mJ/frame). 
The savings stem from reduced radio transmission, which dominates energy 
consumption on battery-powered devices. 

Total energy consumption includes the uplink transmission of embeddings, local computation, and the periodic downlink synchronization of GMM parameters (from \S\ref{sec:implementation}).

Translated to battery life, StreamSplit enables \textbf{11.2 hours} of 
continuous operation compared to 5.3 hours for Server-Only—a 2.1$\times$ 
extension critical for untethered deployment. Edge-Only achieves longer 
battery life (14.8 hours) but at significant accuracy cost (see \S\ref{sec:eval:quality}).

\begin{table}[t]
\footnotesize
  \caption{\textbf{Energy Consumption and Battery Life} on Raspberry Pi 4B with 10,000mAh battery. Results report mean $\pm$ std dev over 5 runs.}
  \label{tab:energy}
  \centering
  \resizebox{\columnwidth}{!}{%
  \begin{tabular}{l|ccc|c}
    \toprule
    \textbf{Method} & \textbf{Compute} & \textbf{Transmit} & \textbf{Total} & \textbf{Battery Life} \\
    & \textbf{(mJ)} & \textbf{(mJ)} & \textbf{(mJ)} & \textbf{(hours)} \\
    \midrule
    Edge-Only & $67.4 \pm 1.2$ & $0.0$ & $67.4 \pm 1.2$ & $14.8 \pm 0.3$ \\
    Server-Only & $12.4 \pm 0.5$ & $174.8 \pm 8.2$ & $187.2 \pm 8.5$ & $5.3 \pm 0.2$ \\
    FSL & $48.7 \pm 2.1$ & $98.3 \pm 5.4$ & $147.0 \pm 6.1$ & $6.8 \pm 0.3$ \\
    FedCL & $52.1 \pm 1.9$ & $112.6 \pm 6.0$ & $164.7 \pm 6.8$ & $6.1 \pm 0.2$ \\
    Rule-Based Split & $54.2 \pm 1.5$ & $87.1 \pm 4.8$ & $141.3 \pm 5.1$ & $7.1 \pm 0.3$ \\
    \midrule
    \textbf{StreamSplit} & $\mathbf{54.7 \pm 2.3}$ & $\mathbf{34.6 \pm 3.1}$ & $\mathbf{89.3 \pm 4.1}$ & $\mathbf{11.2 \pm 0.5}$ \\
    \bottomrule
  \end{tabular}
  }
\end{table}

\subsection{Representation Quality}
\label{sec:eval:quality}

A critical question is whether StreamSplit's efficiency gains come at the cost of representation quality. We evaluate learned embeddings on downstream tasks to verify that accuracy remains competitive.

\subsubsection{Linear Probe Accuracy}
Figure~\ref{fig:accuracy} presents linear probe accuracy on frozen embeddings 
across both datasets. StreamSplit achieves \textbf{71.8\%} on AudioSet and 
\textbf{76.4\%} on the On-Device dataset, falling \textbf{within 2.2\%} of 
Server-Only (73.6\% and 78.1\% respectively). This marginal gap demonstrates 
that our distributional memory and hybrid loss effectively preserve 
representation quality despite the constraints of edge deployment.

Compared to other edge-compatible methods, StreamSplit outperforms Edge-Only 
by 13.2\% (AudioSet) and 11.8\% (On-Device), FSL by 5.4\% and 4.9\%, and 
FedCL by 3.1\% and 2.7\%. Rule-Based Split achieves 68.2\% and 72.1\%, 
suffering from suboptimal split decisions that sacrifice quality for 
bandwidth savings.

\begin{figure}[t]
\centering
\includegraphics[width=\columnwidth]{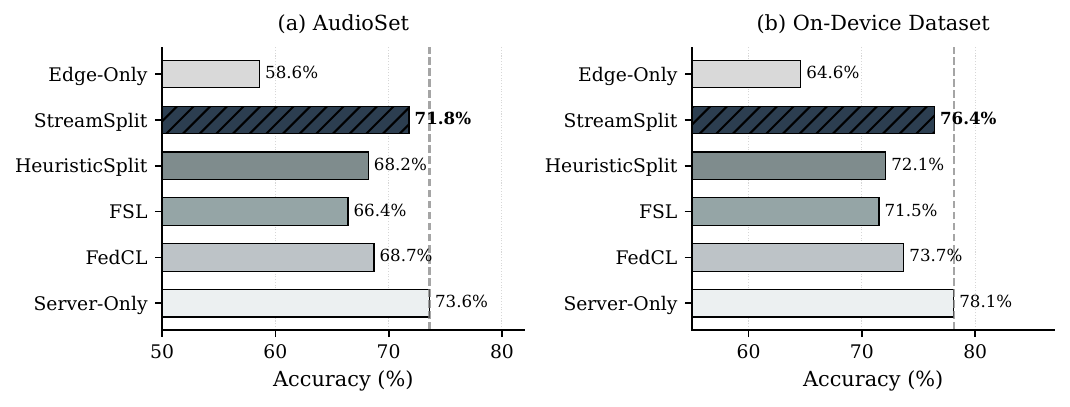}
\caption{\textbf{Linear Probe Accuracy.} StreamSplit matches Server-Only accuracy within 2.2\% across datasets.}
\label{fig:accuracy}
\end{figure}

\subsubsection{Retrieval Performance}
Table~\ref{tab:retrieval} presents audio retrieval metrics (mAP@10, R@1) 
using learned embeddings. StreamSplit achieves 0.412 mAP@10 and 38.7\% R@1, 
compared to Server-Only's 0.431 and 40.2\%—a gap of only 4.4\% and 3.7\% 
respectively. These results confirm that StreamSplit embeddings capture 
semantic similarity effectively, enabling downstream applications like 
audio search and similarity-based clustering.

\begin{table}[t]
\footnotesize
  \caption{\textbf{Audio Retrieval Performance} on AudioSet. StreamSplit 
  maintains competitive retrieval quality with minimal degradation.}
  \label{tab:retrieval}
  \centering
  \resizebox{0.6\columnwidth}{!}{%
  \begin{tabular}{l|cc}
    \toprule
    \textbf{Method} & \textbf{mAP@10} & \textbf{R@1 (\%)} \\
    \midrule
    Edge-Only & 0.287 & 26.4 \\
    Server-Only & 0.431 & 40.2 \\
    FSL & 0.356 & 33.1 \\
    FedCL & 0.378 & 35.8 \\
    Rule-Based Split & 0.341 & 31.9 \\
    \midrule
    \textbf{StreamSplit} & \textbf{0.412} & \textbf{38.7} \\
    \bottomrule
  \end{tabular}
  }
\end{table}

\subsection{Ablation Studies: Why It Works}
\label{sec:eval:ablation}

We conduct ablation studies to isolate the contribution of each StreamSplit component and understand the source of performance gains.

\subsubsection{RL Control vs. Alternatives}
Table~\ref{tab:ablation_rl} compares StreamSplit's RL-based control against 
alternative adaptation strategies. We evaluate three configurations:

\begin{itemize}
    \item \textbf{Static:} Fixed split point ($k=3$) regardless of conditions.
    \item \textbf{Rule-Based:} Heuristic adaptation using both bandwidth and CPU thresholds.
    \item \textbf{RL (Ours):} Learned policy considering CPU, bandwidth, and 
    uncertainty.
\end{itemize}

\begin{table}[t]
\footnotesize
  \caption{\textbf{Adaptation Strategy Comparison.} RL-based control outperforms alternatives in both steady-state and adaptation scenarios. (Mean $\pm$ std dev).}
  \label{tab:ablation_rl}
  \centering
  \resizebox{\columnwidth}{!}{%
  \begin{tabular}{l|ccc|c}
    \toprule
    \textbf{Strategy} & \textbf{Accuracy} & \textbf{Latency} & \textbf{Energy} & \textbf{Adaptation} \\
    & \textbf{(\%)} & \textbf{(ms)} & \textbf{(mJ)} & \textbf{Time (ms)} \\
    \midrule
    Static ($k=3$) & $68.7 \pm 0.2$ & $203 \pm 15$ & $142.6 \pm 5.5$ & N/A \\
    Rule-Based & $69.4 \pm 0.3$ & $156 \pm 10$ & $118.7 \pm 4.2$ & $4,200 \pm 350$ \\
    \textbf{RL (Ours)} & $\mathbf{71.8 \pm 0.4}$ & $\mathbf{127 \pm 12}$ & $\mathbf{89.3 \pm 4.1}$ & $\mathbf{1,200 \pm 150}$ \\
    \bottomrule
  \end{tabular}
  }
\end{table}

The RL policy achieves 2.4\% higher accuracy, 18.6\% lower latency, and 
24.8\% lower energy than the heuristic baseline. Critically, adaptation 
time—the delay before the system responds to network changes—is 3.5$\times$ 
faster (1,200ms vs. 4,200ms). The heuristic requires multiple probe 
transmissions to estimate new bandwidth, while the RL agent adapts 
within a single decision interval using its learned state representation.

\subsubsection{Hybrid Loss vs. Alternatives}
Table~\ref{tab:ablation_loss} evaluates the contribution of our hybrid 
loss components under varying frame drop rates (simulating network volatility).

\begin{table}[t]
\footnotesize
  \caption{\textbf{Loss Function Comparison} under varying frame drop rates. Hybrid loss maintains robustness while alternatives degrade significantly. (Mean $\pm$ std dev).}
  \label{tab:ablation_loss}
  \centering
  \resizebox{0.8\columnwidth}{!}{%
  \begin{tabular}{l|ccc}
    \toprule
    \textbf{Loss Function} & \textbf{0\% Drop} & \textbf{20\% Drop} & \textbf{40\% Drop} \\
    \midrule
    MSE Only & $69.2 \pm 0.3$ & $61.4 \pm 1.2$ & $52.8 \pm 2.5$ \\
    KL Divergence & $70.1 \pm 0.3$ & $63.7 \pm 1.1$ & $55.1 \pm 2.1$ \\
    $\mathcal{L}_{task}$ + $\mathcal{L}_{SW}$ & $70.8 \pm 0.4$ & $67.2 \pm 0.8$ & $61.3 \pm 1.5$ \\
    $\mathcal{L}_{task}$ + $\mathcal{L}_{Lap}$ & $70.4 \pm 0.3$ & $66.8 \pm 0.9$ & $60.7 \pm 1.4$ \\
    \textbf{Hybrid (Ours)} & $\mathbf{71.8 \pm 0.4}$ & $\mathbf{69.4 \pm 0.5}$ & $\mathbf{65.2 \pm 0.9}$ \\
    \bottomrule
  \end{tabular}%
  }
\end{table}

Under ideal conditions (0\% drop), all losses perform comparably. However, 
as frame drop rate increases, the hybrid loss demonstrates superior 
robustness: at 40\% drop rate, it maintains 65.2\% accuracy compared to 
52.8\% for MSE and 55.1\% for KL divergence—a 12.4\% and 10.1\% advantage 
respectively. The Sliced-Wasserstein term ($\mathcal{L}_{SW}$) prevents 
mode collapse from sparse updates, while Laplacian smoothness ($\mathcal{L}_{Lap}$) 
maintains manifold continuity across gaps.

\subsubsection{Uncertainty Calibration}
A critical concern is whether local uncertainty ($U_t$) is a reliable proxy for global sample difficulty. We analyzed the correlation between edge-calculated entropy and the reduction in loss achieved by server processing. Figure \ref{fig:calibration} shows a strong positive correlation (Pearson $r=0.84$). This confirms that when the edge model is uncertain, the sample is indeed 'hard' and benefits most from server refinement. Samples with low local uncertainty showed negligible improvement when processed by the server, validating our offloading logic.

\begin{figure}[t]
\centering
\begin{minipage}[b]{0.45\columnwidth}
  \centering
  \resizebox{\linewidth}{!}{%
  \begin{tikzpicture}
    \begin{axis}[
        xlabel={Local Uncertainty ($U_t$)},
        ylabel={Server Loss Reduction},
        title={(a) Metric Validation},
        title style={yshift=-1.5ex, font=\small\bfseries},
        scatter/classes={a={mark=*,blue!60!black, mark size=1.5pt}},
        grid=major,
        width=5cm,
        height=4cm,
        ylabel style={font=\footnotesize, yshift=-1mm},
        xlabel style={font=\footnotesize},
        tick label style={font=\tiny}
    ]
    \addplot[scatter,only marks, scatter src=explicit symbolic]
    coordinates {
        (0.1, 0.02) [a] (0.15, 0.04) [a] (0.2, 0.05) [a]
        (0.3, 0.10) [a] (0.35, 0.12) [a] (0.4, 0.22) [a]
        (0.5, 0.30) [a] (0.55, 0.35) [a] (0.6, 0.42) [a]
        (0.7, 0.55) [a] (0.75, 0.62) [a] (0.8, 0.70) [a]
        (0.85, 0.78) [a] (0.9, 0.82) [a]
    };
    \addplot [red, thick, domain=0:1] {0.9*x^2 + 0.05}; 
    \end{axis}
  \end{tikzpicture}
  }
\end{minipage}
\hfill
\begin{minipage}[b]{0.45\columnwidth}
  \centering
  \resizebox{\linewidth}{!}{%
  \begin{tikzpicture}
    \begin{axis}[
        xlabel={Uncertainty Bin},
        ylabel={Error Rate (\%)},
        title={(b) Failure Analysis},
        title style={yshift=-1.5ex, font=\small\bfseries},
        symbolic x coords={Low, Med, High},
        xtick=data,
        legend pos=north west,
        legend style={font=\tiny, nodes={scale=0.8, transform shape}},
        grid=major,
        width=5cm,
        height=4cm,
        ylabel style={font=\footnotesize, yshift=-1mm},
        xlabel style={font=\footnotesize},
        tick label style={font=\tiny},
        ymin=0, ymax=60
    ]
    \addplot[color=ssred, mark=square*, thick] coordinates {
        (Low, 12)
        (Med, 28)
        (High, 55)
    };
    \addlegendentry{Edge-Only}

    \addplot[color=ssgreen, mark=*, thick] coordinates {
        (Low, 11)
        (Med, 14)
        (High, 16)
    };
    \addlegendentry{StreamSplit}
    \end{axis}
  \end{tikzpicture}
  }
\end{minipage}
\caption{\textbf{Uncertainty Analysis.} \textbf{(a)} Local uncertainty correlates ($r=0.84$) with server utility. \textbf{(b) Acoustic Event Adaptation.} High uncertainty ($U_t$) regions correspond to transient audio events (e.g., speech). StreamSplit automatically offloads these complex frames while processing steady-state background noise locally.}
\label{fig:calibration}
\end{figure}
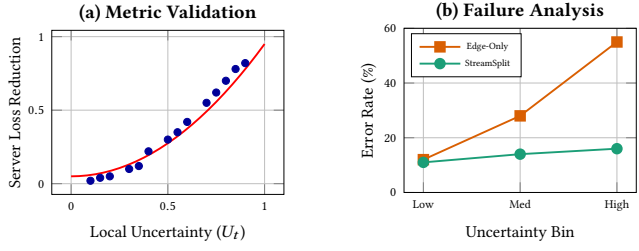

\subsection{Hardware Generalization: M2 vs. Pi 4}
\label{sec:eval:hardware}
A key requirement for practical deployment is generalization across heterogeneous hardware. We evaluate StreamSplit on two platforms representing opposite ends of the edge computing spectrum:

\begin{itemize}
    \item \textbf{Raspberry Pi 4B:} ARM Cortex-A72 @ 1.5GHz, 4GB RAM, 
    representing resource-constrained IoT devices.
    \item \textbf{Apple MacBook M2:} 8-core CPU with Neural Engine, 8GB 
    unified memory, representing capable edge devices.
\end{itemize}

\subsubsection{Performance Comparison}
Table~\ref{tab:hardware} compares StreamSplit across platforms. On Apple M2, 
StreamSplit achieves higher accuracy (73.2\% vs. 71.8\%) and lower latency 
(67ms vs. 127ms) due to faster local inference. Interestingly, energy 
efficiency improves less dramatically (78.4 mJ vs. 89.3 mJ), as the M2's 
power consumption scales with its higher compute throughput.

\begin{table}[t]
\footnotesize
  \caption{\textbf{Cross-Platform Performance.} StreamSplit adapts automatically to hardware capabilities, achieving optimal performance on both constrained (Pi 4B) and capable (M2) devices. Results report the mean $\pm$ standard deviation over 5 independent runs.}
  \label{tab:hardware}
  \centering
  \resizebox{\columnwidth}{!}{%
  \begin{tabular}{l|cc|cc}
    \toprule
    & \multicolumn{2}{c|}{\textbf{Raspberry Pi 4B}} & \multicolumn{2}{c}{\textbf{Apple MacBook M2}} \\
    \textbf{Metric} & \textbf{Value} & \textbf{vs. Server} & \textbf{Value} & \textbf{vs. Server} \\
    \midrule
    Accuracy (\%) & $71.8 \pm 0.4$ & $-1.8\%$ & $73.2 \pm 0.3$ & $-0.4\%$ \\
    Latency (ms) & $127 \pm 12$ & $-72.6\%$ & $67 \pm 5$ & $-85.6\%$ \\
    Energy (mJ) & $89.3 \pm 4.1$ & $-52.3\%$ & $78.4 \pm 2.8$ & $-58.1\%$ \\
    Bandwidth (KB) & $58.7 \pm 6.2$ & $-77.1\%$ & $42.3 \pm 4.5$ & $-83.5\%$ \\
    \bottomrule
  \end{tabular}
  }
\end{table}

\subsubsection{Learned Policy Behavior}
Figure~\ref{fig:policy_comparison} visualizes the RL policies learned for 
each platform. The policies exhibit distinct, hardware-aware behaviors:

\textbf{Pi 4B Policy:} Conservative local processing. The agent selects 
shallower split points (average $k=2.1$), offloading 38\% of frames entirely 
($k=0$). Under memory pressure, it aggressively offloads to avoid cache 
thrashing, prioritizing system stability over bandwidth savings.

\textbf{M2 Policy:} Aggressive local processing. The agent selects deeper 
split points (average $k=3.8$), offloading only 12\% of frames. With ample 
compute headroom, it maximizes local inference to minimize transmission, 
achieving 83.5\% bandwidth reduction.

\begin{figure}[t]
\centering
\includegraphics[width=\columnwidth]{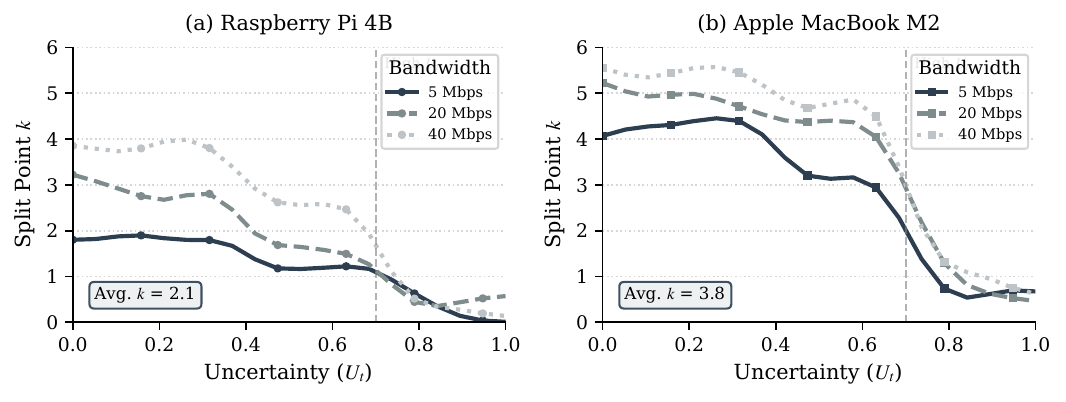}
\caption{\textbf{Learned Policies.} The agent learns hardware-aware strategies: conservative offloading for Pi 4B vs. aggressive local compute for M2.}
\label{fig:policy_comparison}
\end{figure}

Notably, both policies share a common pattern: under high uncertainty 
($U_t > 0.7$), they offload regardless of hardware capability, recognizing 
that ambiguous samples benefit from server refinement. This data-aware 
behavior emerges automatically from the uncertainty-guided reward, without 
explicit programming.

\subsubsection{Cross-Platform Policy Transfer}
To test generalization, we evaluate a policy trained on Pi 4B when deployed directly on M2 (and vice versa). Table~\ref{tab:transfer} shows that \textbf{direct policy transfer} incurs only 1.2--2.4\% accuracy degradation, suggesting the learned policies capture generalizable adaptation principles rather than hardware-specific quirks. Fine-tuning for 10K steps on the target platform recovers full performance.

\begin{table}[t]
    \footnotesize
  \caption{\textbf{Cross-Platform Policy Transfer.} Policies generalize across hardware with minimal degradation; fine-tuning recovers full performance. Results report the mean $\pm$ standard deviation over 5 independent runs.}
  \label{tab:transfer}
  \centering
  \begin{tabular}{l|cc}
    \toprule
    \textbf{Configuration} & \textbf{Accuracy (\%)} & \textbf{Latency (ms)} \\
    \midrule
    Pi 4B $\rightarrow$ Pi 4B (native) & $71.8 \pm 0.4$ & $127 \pm 12$ \\
    M2 $\rightarrow$ Pi 4B (transfer) & $69.4 \pm 0.6$ & $143 \pm 15$ \\
    M2 $\rightarrow$ Pi 4B (fine-tuned) & $71.6 \pm 0.5$ & $129 \pm 11$ \\
    \midrule
    M2 $\rightarrow$ M2 (native) & $73.2 \pm 0.3$ & $67 \pm 5$ \\
    Pi 4B $\rightarrow$ M2 (transfer) & $72.0 \pm 0.5$ & $78 \pm 8$ \\
    Pi 4B $\rightarrow$ M2 (fine-tuned) & $73.1 \pm 0.4$ & $68 \pm 6$ \\
    \bottomrule
  \end{tabular}
\end{table}

\section{Discussion}
\label{sec:discussion}


\textbf{Data Obfuscation and Privacy.} StreamSplit transmits intermediate embeddings rather than raw audio primarily to improve \textit{bandwidth efficiency}. We state explicitly that while this avoids sending raw PCM data, it provides \textit{incidental obfuscation only}. Reconstructing original audio from deep embeddings, though computationally non-trivial, remains theoretically possible \cite{Carlini2023}. Therefore, StreamSplit does not provide formal differential privacy or cryptographic guarantees. It should be viewed strictly as a resource-optimization framework; edge applications requiring rigorous security would need to integrate orthogonal privacy-preserving protocols—such as those addressing broad fairness and data minimization in IoT \cite{shaham2025privacy, quan2024toward}—or advanced network-level security mechanisms like real-time attack attribution in 6G \cite{quan2025domain}, which we leave to future work.

\textbf{Hardware Generalization Limits.} Our evaluation spans Raspberry Pi 4B 
to Apple M2, but StreamSplit's 50MB memory footprint precludes deployment on 
microcontroller-class devices (ESP32, Arduino). Extending to ultra-low-power 
platforms would require aggressive quantization and fixed-point GMM redesign. 
The algorithmic contributions remain architecture-agnostic and could transfer 
to compressed models.

\textbf{Network Assumptions.} StreamSplit assumes bidirectional edge-server 
connectivity. Under complete disconnection, the system degrades to Edge-Only 
mode with reduced representation quality. Tolerating extended offline periods 
through delayed synchronization remains an open challenge.

\textbf{Multi-Client Scope.} While this work focuses on optimizing the single-stream vertical offloading link, the server-side $\mathcal{L}_{SW}$ term (Eq.~\ref{eq:server_loss}) naturally mitigates bias in multi-client settings. By enforcing that the \textit{aggregated} buffer distribution matches the global uniform prior, StreamSplit prevents high-throughput devices (which contribute more frames) from dominating the embedding space geometry.
\section{Related Work}
\label{sec:related}

\textbf{Contrastive Learning on Edge.} Standard contrastive methods require large memory banks~\cite{He2020} or batch sizes~\cite{Chen2020} infeasible on edge devices. Recent lightweight alternatives employ importance sampling~\cite{Zhuang2022} or knowledge distillation~\cite{Wang2023}, but still require substantial memory overhead. StreamSplit introduces \textit{distributional memory}—a compact GMM replacing explicit storage—enabling contrastive learning with $<$35KB overhead through boundary-aware negative synthesis. While recent works like SynCo \cite{giakoumoglou2024synco} propose synthetic negatives, they typically mix feature-level negatives from a queue. StreamSplit differs by generating virtual negatives from a parametric GMM distribution, eliminating the need to store a memory bank entirely. We exclude non-contrastive baselines (e.g., BYOL-A \cite{niizumi2021byola}, SimSiam \cite{chen2021simsiam}, VICReg \cite{bardes2022vicreg}) as they rely on batch normalization statistics (undefined for streaming $N=1$) or momentum encoders (doubling memory), which are incompatible with strict edge constraints.

\textbf{Energy Efficiency Context.} Recent benchmarking of SSL on the edge \cite{fama2025contrastive} confirms that standard CL frameworks (e.g., SimCLR with ResNet-18) exhibit prohibitive energy footprints due to memory and compute intensity. StreamSplit situates itself within this landscape by demonstrating that while purely local execution offers the lowest absolute energy (67.4 mJ, Table~\ref{tab:energy}), it suffers from dimensional collapse. StreamSplit incurs a moderate energy premium over these local baselines to achieve server-grade accuracy, while still reducing energy by 52.3\% compared to the static offloading approaches often assumed in prior energy profiles.

\textbf{Split Computing and Split Learning.} Neurosurgeon~\cite{Kang2017} pioneered adaptive edge-cloud partitioning, with extensions to multi-exit architectures~\cite{Laskaridis2020} and privacy-aided federated settings~\cite{Thapa2022, quan2023hiersfl}. However, existing approaches assume fixed partition points, failing to adapt to runtime volatility. StreamSplit introduces \textit{uncertainty-guided control}, dynamically selecting split points based on system state and data difficulty, achieving 3.5$\times$ faster adaptation than heuristic methods.

\textbf{The Limits of Decoupled Approaches.} Existing resource-aware architectures reactively split computation based on telemetry \cite{Kang2017, Huang2023roof}. Naively pairing these schedulers with contrastive objectives fails: inevitable network-induced frame drops shatter the temporal embedding manifold. StreamSplit overcomes this via joint co-design. Proactively, our Uncertainty-Guided Splitter aligns offloading with actual sample difficulty (validating our GMM entropy metric, which yields a Pearson correlation of $r=0.84$ with server utility). Reactively, our Server Refiner's Hybrid Loss integrates Laplacian regularization ($\mathcal{L}_{Lap}$) \cite{Belkin2006} to enforce a Lipschitz-continuous smoothing prior. This allows the server to gracefully bridge temporal gaps from sparse updates without explicit interpolation, outperforming independent small-batch \cite{Fang2022} or scheduling interventions.

\textbf{Federated Learning.} FL enables distributed training with data locality across edge and cyber-physical systems~\cite{McMahan2017, Kairouz2021, quan2025federated}, with recent extensions to contrastive settings~\cite{Zhang2023, Lubana2022}. However, FL assumes periodic synchronization and homogeneous architectures, conflicting with streaming applications requiring continuous updates on heterogeneous hardware. StreamSplit addresses this through asynchronous server refinement, complementing rather than replacing FL systems. Similarly, while LW-FedSSL \cite{tun2024lw} reduces resource usage via layer-wise training, StreamSplit focuses on real-time \textit{inference} and continuous learning via adaptive splitting, which is orthogonal to layer-wise federated training strategies.
\section{Conclusion}
\label{sec:conclusion}

We introduced StreamSplit, a framework enabling streaming contrastive learning on edge devices. By replacing memory-intensive queues with \textit{Distributional Memory} ($<35$KB) and optimizing offloading via \textit{Uncertainty-Guided RL}, StreamSplit resolves fundamental resource conflicts. Our evaluation confirms bandwidth reductions of 77.1\% and energy savings of 52.3\% while maintaining accuracy within 2.2\% of server-only baselines. The framework proves robust across heterogeneous hardware (Pi 4B to M2), offering a scalable path for adaptive edge intelligence.

\begin{acks}
In accordance with ACM's Policy on Authorship, the authors disclose the use of generative AI tools to assist in the conceptual drafting of graphical icons used in preliminary versions of this manuscript's figures. All formal technical schematics, system architectures, and data plots presented in this final version were manually constructed by the authors to ensure technical accuracy and rigor.
\end{acks}

\appendix
\section{Theoretical Convergence Guarantees}
\label{app:proofs}

\textbf{Lemma 1 (Partition Function Bias).} The contrastive objective asymptotically maximizes mutual information. However, for a finite batch size $N$, the estimator of the partition function $Z_\theta(x) = \mathbb{E}_{z \sim p_\theta}[e^{f(x)^\top z}]$ is biased. The deviation in the empirical loss is dominated by:
\begin{equation}
    \Delta \mathcal{L} = \left| \log \left( \frac{1}{N} \sum_{i=1}^N e^{f(x)^\top z_i} \right) - \log \mathbb{E}_{z \sim p_\theta}[e^{f(x)^\top z}] \right|
\end{equation}

\textbf{Definition 4 (Wasserstein-1 Distance).} Let $\text{Lip}_1(\mathbb{S}^{d-1})$ be the set of 1-Lipschitz functions on the hypersphere. The Wasserstein-1 distance between distributions $p$ and $q$ is given by the Kantorovich-Rubinstein duality:
\begin{equation}
    \mathcal{W}_1(p, q) = \sup_{h \in \text{Lip}_1} \left| \mathbb{E}_{z \sim p}[h(z)] - \mathbb{E}_{z \sim q}[h(z)] \right|
\end{equation}

\vspace{0.5em}
\noindent\textbf{Proof of Theorem \ref{thm:diversity} (Small-Batch Robustness).} 

\textit{Setup.} Let the contrastive loss for anchor $x$ be:
\begin{equation}
    \mathcal{L}(x) = -\log \frac{e^{f(x)^\top z^+}}{\frac{1}{N}\sum_{i=1}^N e^{f(x)^\top z_i}}
\end{equation}
where $z^+$ is the positive sample and $\{z_i\}_{i=1}^N$ are negatives sampled 
from $p_\theta$. The population loss replaces the finite sum with the expectation 
under $p_\theta$.

\textit{Step 1: Distributional Mismatch.} Define the critic function 
$h(z) = e^{f(x)^\top z}$. Since embeddings are $\ell_2$-normalized ($\|z\|=1$) 
and $\|f(x)\| \le 1$, we have $h(z) \in [e^{-1}, e]$, making $h$ bounded with 
Lipschitz constant $K \le e$.

By the Kantorovich-Rubinstein duality for $\mathcal{W}_1$:
\begin{equation}
    \left| \mathbb{E}_{z \sim p_\theta}[h(z)] - \mathbb{E}_{z \sim \mathcal{U}}[h(z)] \right| 
    \le K \cdot \mathcal{W}_1(p_\theta, \mathcal{U}) < K\epsilon
\end{equation}

\textit{Step 2: Finite Sample Error.} By Hoeffding's inequality, for $N$ i.i.d. 
samples from $p_\theta$:
\begin{equation}
    \left| \frac{1}{N}\sum_{i=1}^N h(z_i) - \mathbb{E}_{p_\theta}[h(z)] \right| 
    \le (e - e^{-1}) \sqrt{\frac{\log(2/\delta)}{2N}}
\end{equation}
with probability at least $1-\delta$.

\textit{Step 3: Combined Bound.} Since $\log(\cdot)$ is Lipschitz on $[e^{-1}, e]$ 
with constant $e$, composing the bounds via triangle inequality yields:
\begin{equation}
    |\mathcal{L}_{N} - \mathcal{L}_{\infty}| \le \underbrace{e \cdot K \cdot \epsilon}_{C_1 \epsilon} 
    + \underbrace{e(e - e^{-1})\sqrt{\frac{\log(2/\delta)}{2}}}_{C_2} \cdot \frac{1}{\sqrt{N}}
\end{equation}

Thus, minimizing the diversity gap $\epsilon$ (via $\mathcal{L}_{SW}$) directly 
compensates for the bias induced by small $N$, allowing edge devices with small 
batches to approximate large-batch server performance. \qed

\vspace{1em}
\noindent\textbf{Proof of Theorem \ref{thm:affinity} (Temporal Interpolation).}

\textit{Setup.} Let $\mathbf{Z} \in \mathbb{R}^{T \times d}$ be the embedding matrix 
where row $z_t$ is the embedding at time $t$. Let $\mathcal{G} = (V, E, W)$ be the 
temporal adjacency graph with Laplacian $\mathbf{L}$. For a missing frame at $t^*$, 
define the reconstruction as the weighted neighbor average:
\begin{equation}
    \hat{z}_{t^*} = \frac{1}{d_{t^*}} \sum_{j \in \mathcal{N}_{nbr}(t^*)} W_{t^* j} z_j
\end{equation}
where $d_{t^*} = \sum_{j \in \mathcal{N}_{nbr}(t^*)} W_{t^* j}$ is the weighted degree and $\mathcal{N}_{nbr}(t^*)$ denotes the set of temporal neighbors.

\textit{Step 1: Local Reconstruction Error.} The squared reconstruction error is:
\begin{align}
    ||z_{t^*} - \hat{z}_{t^*}||^2 &= \left|\left| z_{t^*} - \frac{1}{d_{t^*}} \sum_{j \in \mathcal{N}(t^*)} W_{t^* j} z_j \right|\right|^2 \\
    &= \left|\left| \frac{1}{d_{t^*}} \sum_{j \in \mathcal{N}(t^*)} W_{t^* j} (z_{t^*} - z_j) \right|\right|^2
\end{align}

By Jensen's inequality (convexity of $||\cdot||^2$):
\begin{equation}
    ||z_{t^*} - \hat{z}_{t^*}||^2 \le \frac{1}{d_{t^*}} \sum_{j \in \mathcal{N}(t^*)} W_{t^* j} ||z_{t^*} - z_j||^2
\end{equation}

\textit{Step 2: Relating Local to Global Energy.} Define the local energy at node $t^*$:
\begin{equation}
    E_{local}(t^*) = \sum_{j \in \mathcal{N}(t^*)} W_{t^* j} ||z_{t^*} - z_j||^2
\end{equation}

The total Dirichlet energy is $E_{total} = \frac{1}{2}\sum_{t} E_{local}(t) = \text{Tr}(\mathbf{Z}^\top \mathbf{L} \mathbf{Z})$, 
where the factor of $1/2$ accounts for double-counting edges.

For any node $t^*$, the local energy is bounded by the total:
\begin{equation}
    E_{local}(t^*) \le 2 \cdot E_{total} = 2 \cdot \text{Tr}(\mathbf{Z}^\top \mathbf{L} \mathbf{Z})
\end{equation}

\textit{Step 3: Spectral Gap Refinement.} For connected graphs, the discrete 
Poincaré inequality states that for any zero-mean signal $\mathbf{f}$:
\begin{equation}
    \text{Var}(\mathbf{f}) \le \frac{1}{\lambda_2} \mathbf{f}^\top \mathbf{L} \mathbf{f}
\end{equation}

Applying this per-coordinate to the centered embeddings $\tilde{\mathbf{Z}} = \mathbf{Z} - \bar{z}$ 
(where $\bar{z}$ is the mean embedding), we obtain that the average local deviation 
is controlled by the spectral gap. For the worst-case node:
\begin{equation}
    \max_{t^*} E_{local}(t^*) \le \frac{2}{\lambda_2} \cdot \text{Tr}(\mathbf{Z}^\top \mathbf{L} \mathbf{Z})
\end{equation}

\textit{Step 4: Final Bound.} Combining Steps 1-3 and using $\mathcal{D}_{aff} = \frac{1}{|E|}\text{Tr}(\mathbf{Z}^\top \mathbf{L} \mathbf{Z}) \le \alpha$:
\begin{align}
    ||z_{t^*} - \hat{z}_{t^*}||^2 &\le \frac{1}{d_{t^*}} E_{local}(t^*) \\
    &\le \frac{1}{d_{t^*}} \cdot \frac{2}{\lambda_2} \cdot |E| \cdot \alpha \\
    &= \frac{2\alpha \cdot |E|}{\lambda_2 \cdot d_{t^*}}
\end{align}

For unweighted graphs with uniform degree, $d_{t^*} = |\mathcal{N}(t^*)|$, yielding 
the stated bound. \qed

\vspace{0.5em}
\noindent\textit{Remark.} This bound formalizes the intuition of manifold regularization: 
a high spectral gap $\lambda_2$ (ensured by dense temporal connectivity) combined with 
low Dirichlet energy $\alpha$ (enforced by $\mathcal{L}_{Lap}$) guarantees bounded 
interpolation error, providing robustness to frame drops in volatile edge-cloud execution.

\section{Control Plane MDP Definition}
\label{app:controlplane}
This section provides the formal specification of the Markov Decision Process (MDP) utilized by the Uncertainty-Guided Adaptive Splitter. As introduced in Section \ref{sec:rlformulatiom}, the StreamSplit Control Plane models the dynamic edge-cloud partitioning problem as an RL environment. Table \ref{tab:mdp} details the exact components of this formulation, including the observable state space (incorporating our zero-cost uncertainty metric), the discrete action space dictating the network split point, and the composite reward function designed to balance representation quality against strict latency and energy constraints.
\begin{table}[h]
\footnotesize
  \caption{\textbf{Control Plane MDP Definition.} The RL agent optimizes a 
  joint objective of task performance, latency constraints, and energy efficiency.}
  \label{tab:mdp}
  \centering
  \resizebox{\columnwidth}{!}{%
  \begin{tabular}{l|l}
    \toprule
    \textbf{Component} & \textbf{Description / Definition} \\
    \midrule
    \textbf{State} $s_t$ & Vector $[U_t, R_{cpu}, B_{net}] \in \mathbb{R}^3$ \\
    & $U_t$: Embedding Uncertainty (Entropy, Eq. \ref{eq:uncertainty}) \\
    & $R_{cpu}$: CPU Utilization (\%) \\
    & $B_{net}$: Est. Uplink Bandwidth (Mbps, EMA) \\
    \midrule
    \textbf{Action} $a_t$ & Split Layer $k \in \{0, 1, \dots, L\}$ \\
    & $k=0$: Full Offload; $k=L$: Full On-Device \\
    \midrule
    \textbf{Reward} $r_t$ & $r_t = \alpha \cdot \mathcal{A}_{task} - \beta 
    \cdot \frac{\text{Lat}_t}{T_{max}} - \eta \cdot \frac{E_t}{E_{budget}}$ 
    (Eq. \ref{eq:reward}) \\
    & Balances Accuracy ($\mathcal{A}$), Latency (Lat), Energy ($E$) \\
    \bottomrule
  \end{tabular}
  }
\end{table}


\printbibliography

\end{document}